\documentclass[11pt]{book}
\usepackage[sectionbib,authoryear]{natbib}

\usepackage{lipsum}
\usepackage[normalem]{ulem}
\usepackage{alltt}
\usepackage[anticlockwise,figuresright]{rotating}
\usepackage{boites}%
\usepackage{boites_exemples}%
\usepackage{wrapfig}
\usepackage{epstopdf}%

\usepackage{ifthen}%
\usepackage{amsmath}%
\usepackage{color}%
\usepackage{xcolor}%
\usepackage{float}%
\usepackage{graphicx}%

\usepackage{enumerate}%
\usepackage{latexsym}
\usepackage[fleqn]{mathtools}
\usepackage{amssymb,amsthm}
\usepackage{mathrsfs,makeidx,listings,verbatim,moreverb}

\theoremstyle{plain}
\newtheorem{theorem}{Theorem}[chapter]

\newtheorem{example}{Example}[chapter]
\newtheorem{definition}{Definition}[chapter]

\usepackage[breaklinks]{hyperref}%
\hypersetup{colorlinks=false,linkbordercolor=white,pdfborder={0 0 0},citebordercolor=white,bookmarksopen=false}%
\usepackage{breakurl}%
\usepackage{enumerate}%
\usepackage{latexsym}
\usepackage[fleqn]{mathtools}
\usepackage{amssymb,amsthm}
\usepackage{mathrsfs,makeidx,listings,verbatim,moreverb}

\usepackage[breaklinks]{hyperref}%
\hypersetup{colorlinks=false,linkbordercolor=white,pdfborder={0 0 0},citebordercolor=white,bookmarksopen=false}%
\usepackage{breakurl}%
\usepackage{fancyhdr}

\usepackage[caption=false]{subfig}
\usepackage{hyperref}
\usepackage[capitalize]{cleveref}
\usepackage{bm}
\usepackage{bbm}

\usepackage{tikz}
\usetikzlibrary{tikzmark,calc,decorations.pathreplacing,patterns}
\usepackage{tikzsymbols}
\usepackage{pgfplots} 
\pgfplotsset{compat=1.16}
\usepgfplotslibrary{fillbetween}

\usepackage{enumitem}
\usepackage{pifont}
\usepackage{multirow}

\newcommand{\txt}[1]{\text{\normalfont #1}}  
\newcommand{\zero}{\textcolor[rgb]{.9,.9,.9}{0}}

\DeclareMathOperator{\T}{\top}
\DeclareMathOperator{\HT}{H}
\DeclareMathOperator{\F}{F}
\DeclareMathOperator{\rank}{rank}

\DeclareMathOperator{\diag}{diag}

\DeclareMathOperator{\vecm}{vec}

\DeclareMathOperator{\kron}{\otimes}
\DeclareMathOperator{\krao}{\odot}

\DeclareMathOperator{\tx}{t}
\DeclareMathOperator{\rx}{r}
\DeclareMathOperator{\txrx}{\tx\!\rx}
\DeclareMathOperator{\xx}{\xi}
\DeclareMathOperator{\separator}{:}

\renewcommand\cite{\citep}
\newcommand\blfootnote[1]{%
	\begingroup
	\renewcommand\thefootnote{}\footnote{#1}%
	\addtocounter{footnote}{-1}%
	\endgroup
}

\makeindex

\begin{document}
	
	\mainmatter
	\tableofcontents
	\thispagestyle{empty}
	\setcounter{chapter}{8}
	
	\setcounter{page}{1}
\chapter{Sparse Sensor Arrays for Active Sensing: Models, Configurations and Applications}

Robin Rajam\"{a}ki and Visa Koivunen (Aalto University, Finland)
%
%
%
%


\section{Introduction}\blfootnote{The work of R.R. was supported by the Ulla Tuominen foundation, and the work of V.K. by the Nokia Foundation's visiting professor grant. Both R.R. and V.K. were also supported by the Finnish Defense Research Agency.}

In active sensing, sensors such as antennas or piezeoelectric transducers probe their surroundings by emitting self-generated energy. The signals backscattered to the receivers are then processed to extract information from the physical environment---typically in the form of target detections and parameter estimates. Classical applications of active sensing include radar, sonar, wireless communications, medical ultrasound. Active sensing also plays a key role in emerging applications such as autonomous sensing \cite{hugler2018radar}, robotics \cite{schouten2019principles}, automotive radar \cite{patole2017automotive,sun2020mimoradar}, and integrated sensing and communications \cite{mishra2019toward,ma2020joint,ahmadipour2022aninformation}.

In contrast to passive (receive-only) sensing, active sensing systems have the flexibility to control the properties of the emitted signals, commonly known as \emph{waveforms}. Waveforms can be adapted to various tasks, such as target search or tracking, for example, by shaping the waveform in the time, frequency, polarization, or angle domain. Indeed, a key benefit of an active \emph{multi-sensor array} is the ability to focus and steer energy in desired directions by beamforming on transmit. This provides coherent combining gain and  spatial selectivity, which improves signal-to-noise-ratio (SNR), enables resolving closely spaced scatterers, reduces interference in and from undesired directions, and increases reliability of communication links.

As previous chapters of this book have demonstrated, \emph{sparse sensor arrays} offer several advantages over conventional uniform array geometries. These benefits include improved resolution using fewer physical sensors, the capability to identify more scatterers than physical sensors, as well as robustness to antenna system non-idealities, such as mutual coupling. In active sensing, much of this is facilitated by the effective transmit-receive virtual array known as the \emph{sum co-array}. The sum co-array consists of the pairwise sums of transmit and receive sensor positions. Properly designed sparse arrays with $N_{\tx}$ physical transmitters and $N_{\rx}$ receivers can therefore achieve a uniform virtual array with on the order of $N_{\tx}N_{\rx}$ virtual sensors---corresponding to an increase by a factor of $N_{\tx}$ compared to a receive-only array. The benefits of a large virtual array are well-known in applications ranging from coherent imaging \cite{hoctor1990theunifying} to multiple-input multiple-output (MIMO) radar \cite{li2007mimoradar}. The sum co-array is analogous to the difference co-array in passive sensing. However, whereas the difference co-array results from computing second-order statistics under the assumption of uncorrelated source signals, the sum co-array requires no such statistical assumptions.

Most works on sparse arrays in the past decade have focused on passive sensing. As a result, active sparse arrays are relatively less explored. Many concepts originally developed for passive sparse arrays can be extended to the active case---e.g., see \cite{hoctor1996arrayredundancy} in case of the Minimum-redundancy array (MRA) \cite{moffet1968minimumredundancy}. Fundamental differences between the two sensing modalities nevertheless give rise to unique, sometimes subtle, design issues. For example, one may or may not have the freedom to independently choose the transmit and receive sensor positions. Hence, care must be taken when defining active array redundancy and the (active) MRA configuration. If transmitters and receivers are not constrained to occupy the same positions, completely non-redundant configurations are widely known. The most notable example is perhaps the nested geometry also known as the MIMO array \cite{li2007mimoradar}. However, less is known about MRAs when the physical transmit and receive arrays are constrained to have the same geometry. Such configurations---regularly used in monostatic radar and medical ultrasound applications---are of interest when the area for placing the physical sensors is tightly constrained. A disadvantage of these MRA geometries is that they are computationally expensive to find. Hence, scalable low-redundancy array geometries need to be developed. A natural question is therefore whether well-known passive array configurations, such as the Nested array \cite{pal2010nested}, could be adapted to active sensing? This chapter provides a positive answer by outlining a general and pragmatic framework for generating symmetric arrays suitable for both active and passive sensing.

Another pertinent question when employing sparse arrays is how to mitigate grating lobes arsing from spatial undersampling. In active sensing, the flexibility of independently controlling both the transmit (Tx) and receive (Rx) beampatterns provides a means to shape the effective Tx-Rx beampattern, which is often of primary interest in sensing, especially when employing linear processing at both the transmitter and receiver. Since this effective beampattern is the product of the Tx and Rx beampatterns, a null in one beampattern can cancel a side lobe in the other. A general model for synthesizing desirable Tx-Rx beampatterns, known as image addition, was developed by Kassam \emph{et al.} \cite{luthra1983sidelobe,hoctor1990theunifying,kozick1991linearimaging} building on prior works on coherent and incoherent imaging in optical holography \cite{gabor1965optical} and radio astronomy \cite{wild1965anew}, respectively. Image addition synthesizes a desired Tx-Rx beampattern by linearly combining so-called component beampatterns obtained using different Tx-Rx beamforming weight pairs. However, when the array configuration contains redundant virtual sensors, there is an opportunity to optimize the beamforming weights. Consequently, this chapter considers the problem of minimizing the number of component beampatterns, and hence the synthesis time of the desired composite beampattern, since each component may correspond to a separate transmission. We also briefly discuss the fundamental spatio-temporal trade-off between sparsity of the array geometry and the number of component weights required to synthesize any feasible Tx-Rx beampattern.

\subsection{Goals, scope, and organization}
The goal of this chapter is to provide an accessible introduction to active sparse array design and signal processing. We attempt to strike a balance between providing a useful resource on active sparse arrays for practitioners, as well as highlighting problems of theoretical interest to researchers in the field. 

We focus on joint transmit-receive beamforming and one-dimensional array configurations with low-redundancy; largely following our two previous articles \cite{rajamaki2020hybrid,rajamaki2021sparsesymmetric}, and the first author's thesis \cite{rajamaki2021sparsesensor}. For brevity, several important topics in active sensing are treated only cursorily. Notable omissions include waveform design, space-time adaptive processing, and high-resolution localization algorithms. For detailed treatments on these topics, see \cite{levanon2004radar,guerci2015spacetime,stoica2005spectral} and references therein.

The chapter is organized as follows. \cref{sec:signal_model} presents the active sensing signal model. \cref{sec:geometry} reviews sparse array configurations for active sensing, with an emphasis on the active MRA, and scalable low-redundancy symmetric sparse arrays. \cref{sec:coarray_beamforming} examines beamforming using sparse arrays. In particular, we consider the problem of synthesizing a desired joint transmit-receive beampattern using multiple component beampatterns. \cref{sec:applications} briefly illustrates example applications of sparse active arrays. Finally, \cref{sec:conclusions} concludes the chapter.

\subsection{Notation}
 We denote matrices by boldface uppercase, e.g, $ \bm{A} $; vectors by boldface lowercase, $ \bm{a} $; and scalars by unbolded letters, $ A,a $. The $(n,m) $th element of matrix $ \bm{A} $ and $\bm{A}_i$ is $ A_{nm} $ and $ [\bm{A}_i]_{n,m} $, respectively. Furthermore, $ \bm{A}^{\T} $, $ \bm{A}^{\HT} $, and $ \bm{A}^\ast $ denote transpose, Hermitian transpose, and complex conjugation. The $N\times N$ identity matrix is denoted by $ \bm{I}_N $, and the standard unit vector, consisting of zeros except for the $ i $th entry, which is unity, by $ \bm{e}_i $ (dimension specified separately). Moreover, $ \txt{vec}(\cdot) $ stacks the columns of its matrix argument into a column vector, whereas $ \txt{diag}(\cdot) $ constructs a diagonal matrix of its vector argument. The ceiling function is denoted by $ \lceil \cdot \rceil $. The Kronecker and Khatri-Rao (columnwise Kronecker) products are denoted by $ \kron $ and $ \krao $, respectively. Sets are denoted by bold blackboard letters, such as $\mathbb{A}$. The maximum element of a finite set of real-valued scalars $\mathbb{A}$ is denoted by $\max\mathbb{A} $; the subset of integers between $a,b\in\mathbb{R}$ by $ [a\separator b]\triangleq \{c\in\mathbb{Z}\ |\ a\leq c\leq b; a,b\in \mathbb{R}\}$; and the sum set of $\mathbb{A}$ and $\mathbb{B}$ by $\mathbb{A}+\mathbb{B}\triangleq\{a+b\ |\ a\in\mathbb{A}, b\in\mathbb{B}\}$. Finally, subscripts `$ {\tx} $', `$ {\rx} $', and `${\xx}$' denote ``transmitter'', ``receiver'', and ``transmitter or receiver'', respectively.

\section{Active sensing signal model}\label{sec:signal_model}

This section introduces the considered active sensing signal model. For simplicity of presentation, we restrict our attention to linear transmit and receive array geometries with fully digital beamforming architectures. This implies a discrete-time model, where each transmit sensor may launch a different waveform, and the output of each receive sensor is available. For extensions to planar arrays and hybrid or analog beamforming, see \cite{rajamaki2021sparsesensor}.

\subsection{Physical array model}\label{sec:physical_model}

\cref{fig:signal_model} depicts the considered active sensing MIMO model consisting of co-located transmit (Tx) and receive (Rx) arrays. Co-locatedness means that the distance between the Tx and Rx array is negligible w.r.t. their distance to the scatterers, i.e., both arrays observe the same scatterer angles \cite{li2007mimoradar}. This is in stark contrast to distributed arrays leveraging angular diversity \cite{fishler2006spatial}. The Tx and Rx arrays are assumed to have a one-dimensional co-linear geometry.
\begin{figure}
	\centering
	\includegraphics[width=.7\linewidth]{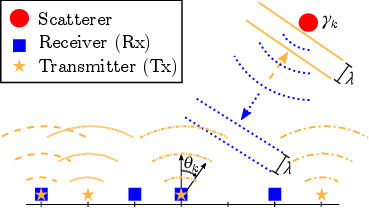}
	\caption{Active sensing signal model. A linear Tx array with $N_{\tx}$ sensors illuminates $K$ far field point scatterers using $N_s\leq N_{\tx}$ linearly independent waveforms. A linear Rx array with $N_{\rx}$ sensors, co-located with the Tx array, measures a superposition of the backscattered planar wavefronts.}
	\label{fig:signal_model}
\end{figure}

The Tx array illuminates a collection of $K$ far field point scatterers, which backscatter the impinging wavefield back at the Rx array. The baseband Tx signals are modulated by a carrier frequency of wavelength $\lambda$ before transmission. The carrier frequency is assumed much larger than the bandwidth of the baseband signal, such that a narrowband signal model approximately holds. We assume that the Tx/Rx sensors lie on a uniform grid of integer multiples of $\lambda/2$ given by set $\mathbb{D}_{\xx}\in\mathbb{Z}$. The number of Tx/Rx sensors is denoted by $N_{\xx}\triangleq |\mathbb{D}_{\xx}|$.

\subsubsection{Angle-delay-Doppler model}
The transmitted signals experience different round-trip delays depending on the distance from the sensing platform to the scatterers, and Doppler shifts due to the relative difference in radial motion between the platform and scatterers. The $k$th scatterer is therefore generally parametrized by scattering coefficient $\gamma_k$, azimuth angle $\theta_k$, delay $\tau_k$, and Doppler shift $\omega_k$. A simple model for the continuous-time baseband Rx array output is \cite{friedlander2012onsignal} 
\begin{align}
	\bm{y}(t)=\sum_{k=1}^K\gamma_k \bm{a}_{\rx}(\theta_k)\bm{a}_{\tx}^{\T}(\theta_k)\bm{s}(t-\tau_k)e^{j\omega_k t}+\bm{n}(t).\label{eq:model_general}
\end{align}
Here $\bm{a}_{\xx}(\theta_k)\in \mathbb{C}^{N_{\xx}}$ denotes the Tx/Rx array manifold vector (recall shorthand `$\xx$' denotes Tx or Rx), and $\bm{s}(t)\in\mathbb{C}^{N_{\tx}}$ is a spatial snapshot of the transmitted waveform at time $t$. Term $\bm{n}(t)\in\mathbb{C}^{N_{\rx}}$ models noise, interference, and clutter. 

To simplify presentation, we will henceforth focus on a angle-only version of \eqref{eq:model_general}. This restriction, commonly employed in for instance MIMO radar \cite{bekkerman2006target,li2007mimoradar}, is justified by the fact that the array geometry, which is of main interest here, primarily affects the sensing performance in the angular domain. For generalizations of \eqref{eq:model_general} to frequency-selective models, see, e.g., \cite[Ch.~5]{li2009mimo}.

\subsubsection{Simplified angle-only model}
Suppose the $K$ scatterers lie in the same delay-Doppler bin, and that the only unknown scatter parameters are the scattering coefficients $\bm{\gamma}\in\mathbb{C}^{K}$ and azimuth angles $\bm{\theta}\in[-\pi/2,\pi/2)^{K}$. Sufficient isolation between delay-Doppler bins can be achieved by proper waveform design and space-time processing. Setting the delay and Doppler parameters to $\tau_k=\omega_k=0$ $\forall k$, and sampling the array output at $T\in\mathbb{N}_+$ time-instances, \eqref{eq:model_general} yields the following $N_{\tx}\times T$ discrete-time model \cite{bekkerman2006target}
\begin{align}
	\bm{Y}
	=\sum_{k=1}^K\gamma_k \bm{a}_{\rx}(\theta_k)\bm{a}_{\tx}^{\T}(\theta_k)\bm{S}^{\T}+\bm{N}
	=\bm{A}_{\rx}\diag(\bm{\gamma})\bm{A}_{\tx}^{\T}\bm{S}^{\T}+\bm{N}.\label{eq:Y}
\end{align}
Here, $\bm{A}_{\xx}=\bm{A}_{\xx}(\bm{\theta})=[\bm{a}_{\xx}(\theta_1),\ldots,\bm{a}_{\xx}(\theta_K)]\in\mathbb{C}^{N_{\xx}\times K}$ is the narrowband Tx/Rx array manifold matrix modeling the relative phase shifts between the sensors induced by the planar wavefronts departing from/impinging on the Tx/Rx array.  The $(n,k)$th entry of $\bm{A}_{\xx}$ depends on the $n$th normalized Tx/Rx sensor position $d_{\xx}[n]\in\mathbb{D}_{\xx}$, $n\in[1\separator N_{\xx}]$, and the $k$th scatterer angle $\theta_k\in[-\pi/2,\pi/2)$, $k\in[1\separator K]$:
\begin{align}
	[\bm{A}_{\xx}]_{n,k}=\exp(j\pi d_{\xx}[n] \sin \theta_k). \label{eq:A_xx}
\end{align}
Moreover, $\bm{S}=[\bm{s}(t_1),\ldots, \bm{s}(t_T)]^{\T}\in\mathbb{C}^{T\times N_{\tx}}$ is a spatio-temporal Tx \emph{waveform matrix}, whose $n$th column corresponds to the waveform transmitted by the $n$th Tx sensor sampled at times $\{t_i\}_{i=1}^T$. For simplicity, we assume the entries of $\bm{N}\in\mathbb{C}^{N_{\rx}\times T}$ are i.i.d. complex-valued zero-mean Gaussian random variables of equal variance $\sigma^2$, i.e., $N_{n,t}\sim\mathcal{CN}(0,\sigma^2)$. For more general clutter and interference models, see, e.g., \cite[Ch.~6]{li2009mimo}.

\cref{eq:Y} can be vectorized and expressed using the Kronecker product $\kron$ and Khatri-Rao (columnwise Kronecker) product $\krao$ as  $\bm{y}\triangleq \vecm(\bm{Y})\in\mathbb{C}^{TN_{\rx}}$:
\begin{align}
	\bm{y}
	=((\bm{S}\bm{A}_{\tx})\krao \bm{A}_{\rx})\bm{\gamma}+\bm{n}
	=(\bm{S}\kron\bm{I}_{N_{\rx}})(\bm{A}_{\tx}\krao \bm{A}_{\rx})\bm{\gamma}+\bm{n}. \label{eq:y_eff}
\end{align}
Here, $\bm{\gamma}\in\mathbb{C}^K$ is the vector of scattering coefficients, $\bm{n}\triangleq\vecm(\bm{N})$ is the vectorized noise matrix, and $\bm{I}_{N_{\rx}}$ is the $N_{\rx}\times N_{\rx}$ identity matrix. \cref{eq:y_eff} follows from the following Kronecker/Khatri-Rao product identities: $\vecm(\bm{H}\diag(\bm{f})\bm{G}^{\T})=(\bm{G}\krao\bm{H})\bm{f}$ and $(\bm{G}\bm{H})\krao(\bm{F}\bm{D})=(\bm{G}\kron\bm{F})(\bm{H}\krao\bm{D})$ for matrices of appropriate dimensions---see \cite[Theorem~2]{liu2008hadamard} for details.

\subsubsection{Waveform matrix}
Waveform matrix $\bm{S}$ crucially impacts the performance of the active sensing system. For example, the choice $\bm{S}$ affects the realized transmit combining gain, the identifiability of the unknown model parameters $\bm{\theta}$ and $\bm{\gamma}$ in \eqref{eq:y_eff}, and the fidelity to which they can be recovered in the presence of noise and interference. While \cref{sec:coarray_beamforming} elaborates on the role of $\bm{S}$ in beamforming, here we only briefly discuss some basic properties of $\bm{S}$. The reader interested in waveform design is referred to the rich literature on the topic; for example, see \cite{bell1993information,levanon2004radar,pillai2011waveform,he2012waveform}.

\paragraph{Waveform rank}
A fundamental property of $\bm{S}$ is its rank, $N_s\triangleq \rank(\bm{S})$. The canonical example of a \emph{unit} rank system ($N_s=1$) is single-input-multiple-output (SIMO)---also known as phased array---radar \cite[Ch.~8]{skolnik1981introduction}, which corresponds to transmitting a single appropriately phased shifted (delayed) and possibly scaled waveform from each Tx sensor. In contrast, full waveform rank ($N_s=N_{\tx}$) is usually employed in MIMO radar \cite{bliss2003multiple,li2007mimoradar}, which transmits linearly independent (typically orthogonal) waveforms from each of the $N_{\tx}$ transmitters. Typically, orthogonal waveforms are used when the radar is in a target search mode. An intermediate waveform rank ($N_s\in[2\separator N_{\tx}-1]$) usually trades off field of view for coherent transmit combining gain \cite{aittomaki2007signal,fuhrmann2010signaling,hassanien2010phasedmimo}. The appropriate choice of $N_s$ is therefore task-dependent. For example, $N_s=1$ may be desirable when focusing energy on single target for tracking purposes, whereas $N_s=N_{\tx}$ is suitable when searching for new targets without \emph{a priori} knowledge of their directions. Increasing $N_s$ also facilitates resolving more targets simultaneously, extending the set of achievable beampatterns, as well as illuminating broader angular intervals while mitigating interferers.

Waveform rank may also be constrained by hardware costs, power consumption, or the available computational resources. For example, full waveform rank ($N_s=N_{\tx}$) can be infeasible due to the required number of expensive and power hungry radio and intermediate-frequency (RF-IF) front-ends at the transmitter. Waveform rank may also affect the hardware complexity or computational burden of the receiver. For instance, a dedicated filter is required for each transmitted waveform if matched or mismatched filtering is employed.

\paragraph{Practical waveform constraints}
The transmit power of the sensing system is always limited in practice. Typical power constraints include the total power constraint $\|\bm{S}\|_{\F}^2\leq \zeta > 0$ or the per sensor power constraint $\|\bm{S}_{:,n}\|_2^2\leq \zeta/N_{\tx} > 0\ \forall n$. Constant modulus waveforms $|S_{t,n}|=c>0\ \forall (t,n)$ are also often favored because a high peak-to-average power ratio leads to the inefficient use of amplifiers. This is especially the case in radar applications, where high Tx power is needed to achieve a desired SNR and coverage over a large surveillance volume. Another key property of the system determined by the waveform is the delay-Doppler response, which is often desired to closely approximate the ideal thumbtack function.

\subsection{Virtual array model}\label{sec:virtual_model}
This section introduces the canonical virtual array model known as the sum co-array \cite{hoctor1990theunifying}. The sum co-array consists of the pairwise sums of Tx-Rx sensor position pairs, and arises naturally in active sensing. 
\begin{definition}[Sum co-array]
	The sum co-array $\mathbb{D}_{\Sigma}$ is defined as
	\begin{align*}
		\mathbb{D}_\Sigma \triangleq  \mathbb{D}_{\tx}+\mathbb{D}_{\rx} = \{d_{\tx}+d_{\rx}\ |\ d_{\tx}\in\mathbb{D}_{\tx}; d_{\rx}\in\mathbb{D}_{\rx} \},
	\end{align*}
	where $\mathbb{D}_{\tx}$ and $\mathbb{D}_{\rx}$ denote the set of Tx and Rx sensor positions, respectively. The cardinality of the sum co-array is denoted by $N_\Sigma\triangleq |\mathbb{D}_\Sigma|$.
\end{definition}
The sum co-array is closely related to the difference co-array, which emerges under the assumption of uncorrelated source signals and hence finds wide use in passive sensing  (see Chapter~1). Recall that the difference co-array (or difference set) of array configuration $\mathbb{D}$ is the set of pairwise sensor position differences $\mathbb{D}_\Delta\triangleq \mathbb{D}-\mathbb{D}=\{d_1-d_2\ |\ d_1,d_2\in \mathbb{D}\}$. Unlike the difference co-array, the sum co-array does not require any statistical assumptions (on the scattering coefficients) to emerge, but is rather a consequence of the co-locatedness and phase coherence of the Tx and Rx arrays.

Similarly to the physical array, a virtual steering matrix $ \bm{A}_\Sigma\in\mathbb{C}^{N_\Sigma \times K} $ can be associated with the sum co-array. The  $ (\ell,k) $th entry of this matrix is
\begin{align}
[\bm{A}_\Sigma]_{\ell,k} = \exp(j\pi d_\Sigma[\ell] \sin\theta_k).\label{eq:A_Sigma}
\end{align}
Matrix $\bm{A}_\Sigma$ is related to the effective Tx-Rx array steering matrix $\bm{A}_{\tx}\krao\bm{A}_{\rx}$ as
\begin{align}
\bm{A}_{\tx}\krao\bm{A}_{\rx}=\bm{\Upsilon}^{\T}\bm{A}_\Sigma, \label{eq:A_krao_A_Sigma}
\end{align}
where $\bm{\Upsilon}\in\{0,1\}^{N_\Sigma \times N_{\tx}N_{\rx}}$ is a \emph{redundancy pattern matrix} \cite{rajamaki2020hybrid}. This binary matrix characterizes the virtual sensor multiplicities. Specifically, it forms a one-to-many map from the unique elements of the sum co-array to the corresponding physical Tx-Rx sensor pairs.
\begin{definition}[Redundancy pattern]
The $ (\ell,i)$th entry of the binary \emph{redundancy pattern matrix} $ \bm{\Upsilon}\in\{0,1\}^{N_\Sigma\times N_{\tx}N_{\rx}} $ is
\begin{align}
	\Upsilon_{\ell,i}
	\triangleq
	\begin{cases}
		1,&\txt{if }d_{\tx}\big[\lceil i/N_{\rx} \rceil\big]+d_{\rx}\big[i-(\lceil i/N_{\rx}\rceil\!-\!1) N_{\rx}\big]\!=\!d_\Sigma[\ell]\\
		0,&\txt{otherwise.}
	\end{cases}
\label{eq:Upsilon}
\end{align}
Here, $d_{\xx}[n]\in\mathbb{D}_{\xx}$ is the $n$th Tx/Rx sensor position and $d_{\Sigma}[\ell]\in\mathbb{D}_{\Sigma}$ is the $\ell$th sum co-array element position. 
\end{definition}
Matrix $\bm{\Upsilon}$ satisfies by definition $\bm{\Upsilon}\bm{\Upsilon}^{\T}=\diag(\bm{\upsilon}_\Sigma)$, where $\bm{\upsilon}_\Sigma\in\mathbb{N}_+^{N_\Sigma}$ is positive integer-valued vector containing the multiplicities of the unique virtual sensors. We assume without loss of generality that the virtual sensor positions $\mathbb{D}_\Sigma= \{d_\Sigma[\ell]\}_{\ell=1}^{N_\Sigma}$ and physical Tx/Rx sensor positions $\mathbb{D}_{\xx}= \{d_{\xx}[n]\}_{n=1}^{N_{\xx}}$ are indexed in ascending order, with their first sensors at the origin, i.e., 
\begin{align}
		d_{\xx}[1]=0<d_{\xx}[2]<\ldots<d_{\xx}[N_{\xx}] \txt{   and   }
		d_\Sigma[1]=0<d_\Sigma[2]<\ldots<d_\Sigma[N_\Sigma].
\label{eq:ordering}
\end{align}
A uniform sum co-array $\mathbb{D}_\Sigma=[0,N_\Sigma-1]$ is also called \emph{contiguous}. A key advantage of a contiguous co-array is that the co-array manifold $\bm{A}_\Sigma$ is a Vandermonde matrix, which guarantees the identifiability of up to $K\propto N_\Sigma$ scatterers---regardless of their configuration, provided the angles are unique.
\begin{example}[Redundancy pattern of ULA]\label{ex:Upsilon}
	Consider the uniform linear array (ULA) $\mathbb{D}_{\tx}=\mathbb{D}_{\rx}=\{0,1\}$ with $N_{\tx}=N_{\rx}=2$ Tx/Rx sensors. The sum co-array is contiguous and given by $\mathbb{D}_\Sigma=\{0,1,2\}$. Consequently, by \eqref{eq:Upsilon}
	the redundancy pattern matrix $\bm{\Upsilon}$ and multiplicity vector $\bm{\upsilon}_\Sigma=\bm{\Upsilon}\bm{\Upsilon}^{\T}\bm{1}$ are
\begin{align*}
	\bm{\Upsilon}=
	\begin{bmatrix}
		1		  &\zero&\zero &\zero\\
		\zero&1&1	     &\zero\\
		\zero&\zero &\zero  &1\\
	\end{bmatrix}
	\quad \txt{and}\quad
	\bm{\upsilon}_\Sigma =
	\begin{bmatrix}
		1\\
		2\\
		1
	\end{bmatrix}.
\end{align*}
The second row of $\bm{\Upsilon}$ has two nonzero entries, since the second virtual sensor (corresponding to element $1$) can be represented as the sum of the first Tx sensor and the second Rx sensor or vice versa, i.e., $1=0+1=1+0$ .
\end{example}
\cref{eq:Upsilon,eq:A_Sigma} enable expressing the measurement vector in \eqref{eq:y_eff} as
\begin{align}
	\bm{y}
	=\underbrace{(\bm{S}\kron\bm{I})\bm{\Upsilon}^{\T}}_{\triangleq \bm{C}}\bm{A}_\Sigma\bm{\gamma}+\bm{n}
	\label{eq:y_sca}
\end{align}
This reformulation allows re-interpreting \eqref{eq:y_eff} as a virtual array measurement model. Specifically, if the waveform rank is $N_s< N_{\Sigma}/N_{\rx}$, then $\bm{y}$ is a \emph{compressed} measurement of the virtual sensor outputs of the sum co-array. The compression is due to the potentially rank-deficient matrix $\bm{C}\triangleq (\bm{S}\kron\bm{I})\bm{\Upsilon}^{\T}\in\mathbb{C}^{TN_{\rx}\times N_\Sigma}$ satisfying $\rank(\bm{C})\leq \min(N_sN_{\rx},N_\Sigma)$. Note that $\bm{C}$ can be tall, yet column-rank deficient, since long waveforms $T\gg N_{\tx}$ may be employed in practice. This is typically the case when using phase-modulated radar codes and pulse compression. Matrix $\bm{C}$ also reveals that the waveform and array geometry interact via redundancy pattern matrix $\bm{\Upsilon}$. This interpretation and its implications remain to be fully explored in future work.

\section{Sparse array configurations}\label{sec:geometry}

This section presents examples of sparse array configurations for active sensing. \cref{sec:categories} starts by categorizing active array geometries based on the number of shared sensors between the Tx and Rx arrays. We focus on geometries with fully overlapping Tx and Rx arrays (shared transceivers) that maximize the size of the sum co-array (assumed contiguous) for a given number of physical sensors. \cref{sec:mra} discusses the class of Minimum-redundancy arrays, which are optimal in this regard. However, finding these array configurations typically requires solving a combinatorial optimization problem, which is impractical even when the number of physical sensors is well below a hundred. Therefore, \cref{sec:scalable} introduces scalable, but possibly suboptimal configurations, whose sum co-array is nevertheless contiguous and grows order-wise optimally in the number of sensors. As a bonus, these geometries are symmetric by construction, which implies that both the sum and difference co-array is contiguous, and hence the arrays are suitable for both active and passive sensing. For active sparse array design approaches beyond the sum co-array perspective adopted herein, see, e.g., Chapter 10 and references therein.

\subsection{Categorization of array configurations based on overlap between Tx and Rx arrays}\label{sec:categories}

A meaningful definition of the MRA requires considering any constraints imposed on the admissible sensor positions, as in active sensing, one may or may not have the freedom to choose the Rx and Tx sensor positions independently. In particular, the number of distinct physical sensors (either transmitting, receiving or transceiving), $ |\mathbb{D}_{\tx}\cup\mathbb{D}_{\rx}| $, may be less than the total number of transmitters and receivers, $ |\mathbb{D}_{\tx}|+|\mathbb{D}_{\tx}| $. These two quantities, together with the number of shared sensors (transceivers), $|\mathbb{D}_{\tx}+\mathbb{D}_{\rx}|$, satisfy
\begin{align*}
	|\mathbb{D}_{\tx}|+|\mathbb{D}_{\rx}|=|\mathbb{D}_{\tx}\cup\mathbb{D}_{\rx}|+|\mathbb{D}_{\tx}\cap\mathbb{D}_{\rx}|.
\end{align*}
Depending on the fraction of sensors shared by the Tx and Rx arrays, active array configurations naturally fall into one of the following three categories:
\begin{enumerate}[label=(\alph*)]
	\item \emph{Fully overlapping} Tx and Rx arrays, $ \mathbb{D}_{\tx}=\mathbb{D}_{\rx}=\mathbb{D}$\label{i:fully_overlapping}
	\item \emph{Partially overlapping} Tx and Rx arrays $\{0\}\subset \mathbb{D}_{\tx}\cap\mathbb{D}_{\rx} \subset\mathbb{D}_{\tx}\cup\mathbb{D}_{\rx}$\label{i:partially_overlapping}
	\item \emph{Non-overlapping} Tx and Rx arrays, $ \mathbb{D}_{\tx}\cap\mathbb{D}_{\rx} = \{0\} $.\label{i:non-overlapping}
\end{enumerate}
\cref{fig:overlap} illustrates the three cases. We will mainly focus on \labelcref{i:fully_overlapping} herein. Note that the Tx and Rx arrays have at least one common sensor (at the origin) due to the adopted ordering convention in \eqref{eq:ordering}. Although the Tx and Rx arrays obviously need not share any sensors in practice, the normalized Tx-Rx array configuration satisfying \eqref{eq:ordering} will fall into one of the above categories \labelcref{i:fully_overlapping,i:partially_overlapping,i:non-overlapping}.

The total number of distinct sensors satisfies $ |\mathbb{D}_{\tx}\cup\mathbb{D}_{\rx}| \in[\max(|\mathbb{D}_{\tx}|,|\mathbb{D}_{\rx}|)\separator |\mathbb{D}_{\tx}|+|\mathbb{D}_{\rx}|-1]$, which in the fully overlapping and non-overlapping cases simplifies to
\begin{align*}
	|\mathbb{D}_{\tx}\cup\mathbb{D}_{\rx}| = 
	\begin{cases}
		|\mathbb{D}|,& \txt{if } \mathbb{D}_{\tx}=\mathbb{D}_{\rx}=\mathbb{D} \txt{ (fully overlapping)}\\
				|\mathbb{D}_{\tx}|+|\mathbb{D}_{\rx}|-1, & \txt{if } \mathbb{D}_{\tx}\cap\mathbb{D}_{\rx} = \{0\} \txt{ (non-overlapping)}.
	\end{cases}
\end{align*}
Fully overlapping transmitters and receivers are typically employed in pulsed systems, such as phased array radar \cite[p.~5]{skolnik1981introduction} and medical ultrasound \cite{macovski1979ultrasonic}. Non-overlapping Tx and Rx arrays are common in applications such as full-duplex communications \cite{everett2016softnull} or automotive MIMO radar \cite{patole2017automotive} employing continuous wave transmission (or even a bistatic Tx-Rx array configuration). Non-overlapping array configurations have the advantage of achieving a completely nonredundant and contiguous sum co-array, as we will see in \cref{sec:mra_known}. However, the physical apertures of the resulting Tx and Rx arrays can be significantly different, which may be impractical when the physical area for placing the sensors is strictly limited. In contrast, fully overlapping configurations have---by definition---the same physical Tx and Rx array aperture, which is exactly half the virtual aperture of the sum co-array, as illustrated in \cref{fig:overlap}.
\begin{figure}
	\centering
	\newcommand{\Na}{3}   
	\newcommand{\Nb}{3}
	\subfloat[Fully overlapping]{    
		\begin{tikzpicture} 
			\begin{axis}[width=5 cm,height=3 cm,ytick={-1,0,1},yticklabels={$\mathbb{D}_{\Sigma}$,$\mathbb{D}_{\tx}$,$\mathbb{D}_{\rx}$},xmin=-0.2,xmax=\Na*\Nb-1+0.2,ymin=-1.5,ymax=1.5,xtick={0,1,...,\Na*\Nb-1},title style={yshift=0 pt},xticklabel shift = 0 pt,xtick pos=bottom,ytick pos=left,axis line style={draw=none},]
					\addplot[blue,only marks,mark=square*,mark size=2.5] coordinates {
					(0,1)
					(1,1)
					(3,1)
					(4,1)
				};
			\addplot[orange,only marks,mark=star,mark size=3.5,very thick] coordinates {
				(0,0)
				(1,0)
				(3,0)
				(4,0)
			};
				\addplot[only marks,mark=*,mark size=2.5] expression [domain=0:8, samples=9] {-1};
			\end{axis}
		\end{tikzpicture}\label{fig:array_fo}
	}\hfill
	\subfloat[Partially overlapping]{    
		\begin{tikzpicture} 
			\begin{axis}[width=5 cm ,height=3 cm,ytick={},yticklabels={},xmin=-0.2,xmax=\Na*\Nb-1+0.2,ymin=-1.5,ymax=1.5,xtick={0,1,...,\Na*\Nb-1},title style={yshift=0 pt},xticklabel shift = 0 pt,xtick pos=bottom,axis line style={draw=none},ymajorticks=false]
				\addplot[blue,only marks,mark=square*,mark size=2.5] coordinates {
					(0,1)
					(1,1)
					(2,1)
					(4,1)
				};
				\addplot[orange,only marks,mark=star,mark size=3.5,very thick] coordinates {
					(0,0)
					(1,0)
					(3,0)
					(4,0)
				};
				\addplot[only marks,mark=*,mark size=2.5] expression [domain=0:8, samples=9] {-1};
			\end{axis}
		\end{tikzpicture}\label{fig:array_pa}
	}\hfill
\subfloat[Non-overlapping]{    
	\begin{tikzpicture} 
		\begin{axis}[width=5 cm ,height=3 cm,ytick={},yticklabels={},xmin=-0.2,xmax=\Na*\Nb-1+0.2,ymin=-1.5,ymax=1.5,xtick={0,1,...,\Na*\Nb-1},title style={yshift=0 pt},xticklabel shift = 0 pt,xtick pos=bottom,axis line style={draw=none},ymajorticks=false]
			\addplot[blue,only marks,mark=square*,mark size=2.5] expression [domain=0:(\Na-1), samples=\Na] {1};
			\addplot[orange,only marks,mark=star,mark size=3.5,very thick] expression [samples at ={0,\Na,...,\Na*(\Nb-1)}] {0};
			\addplot[only marks,mark=*,mark size=2.5] expression [domain=0:(\Na*\Nb-1), samples=\Na*\Nb] {-1};
		\end{axis}
	\end{tikzpicture}\label{fig:array_no}
}
		\caption{Categorization of (aligned) active array configurations based on the overlap between Tx and Rx sensors. The fully overlapping geometry in \protect\subref{fig:array_fo} achieves a contiguous sum co-array of exactly twice the physical Tx/Rx array aperture. The non-overlapping configuration in \protect\subref{fig:array_no} attains the same co-array using fewer physical sensors at the expense of a larger physical array aperture.}\label{fig:overlap}
\end{figure}
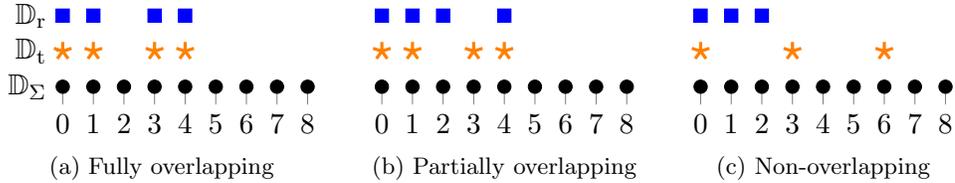

We remark that the extent to which sensors are shared between the Tx and Rx arrays is a physical property of the array. This notion of array overlap is therefore conceptually different from array partitioning or sub-array processing schemes used in, e.g., spatial smoothing \cite{evans1982application}, ESPRIT \cite{roy1989esprit}, multifunction radar \cite{stailey2016multifunction}, full-duplex communications \cite{everett2016softnull}, or joint radar and communications systems \cite{liu2018mumimo}.

\subsection{Minimum-redundancy array (MRA)} \label{sec:mra}
A large contiguous sum co-array is desirable for achieving a high angular resolution and unambiguously resolving as many scatterers as possible. Array processing is also greatly simplified if the sum co-array is contiguous, or if it contains a large contiguous subarray. Hence, it is desirable to maximize the size of the (contiguous) co-array for a given number of physical sensors. The \emph{Minimum-redundancy array} (MRA) \cite{moffet1968minimumredundancy,hoctor1996arrayredundancy} is optimal in this regard, since redundancy quantifies how efficiently the sensors of the physical array are utilized toward constructing a co-array with unique contiguous virtual sensors. Originally defined for passive arrays and the difference co-array \cite{moffet1968minimumredundancy}, the concept of redundancy was later extended to active arrays and the sum co-array \cite{hoctor1996arrayredundancy}. See also \cite{chen2008minimumredundancymimo,weng2011nonuniform} for an extension to the ``difference-sum'' co-array, which arises in active sensing when scatterers are uncorrelated \cite{boudaher2015sparsitybased}. Herein, we focus on the sum co-array-based MRA for active sensing. \cref{sec:redundancy,sec:mra_def} formally define the concepts of redundancy and the MRA, respectively, whereas \cref{sec:mra_known} briefly presents some known MRAs and discusses their structure.

\subsubsection{Redundancy}\label{sec:redundancy}
\emph{Redundancy} quantifies the repetition of the virtual sensors in the co-array. The redundancy of any array configuration with a contiguous sum co-array can be defined as \cite[Eqs.~(7) and (8)]{hoctor1996arrayredundancy}
\begin{align*}
	R(\mathbb{D}_{\tx},\mathbb{D}_{\rx})\triangleq \begin{cases}
		\frac{1}{2}|\mathbb{D}|(|\mathbb{D}|+1)/|\mathbb{D}+\mathbb{D}|,& \txt{if } \mathbb{D}_{\tx} = \mathbb{D}_{\rx}=\mathbb{D}\\
				|\mathbb{D}_{\tx}||\mathbb{D}_{\rx}|/|\mathbb{D}_{\tx}+\mathbb{D}_{\rx}|,& \txt{otherwise.}
	\end{cases}
\end{align*}
An array is said to be \emph{non-redundant} if $ R=1 $, and \emph{redundant} if $ R> 1 $. In the case of fully overlapping Tx and Rx arrays ($\mathbb{D}_{\tx} = \mathbb{D}_{\rx}=\mathbb{D}$), the only two non-redundant array configurations are the trivial arrays $\mathbb{D}=\{0\}$ and $\mathbb{D}=\{0,1\}$ \cite[Corollary~1]{rajamaki2021sparsesymmetric}. A more general definition of redundancy applicable to arrays with a noncontiguous co-array is also possible \cite{moffet1968minimumredundancy,rajamaki2021sparsesensor}. In the active case, this definition can be refined further to account for the exact overlap between the Tx and Rx arrays \cite{rajamaki2021sparsesensor}. This facilitates a more informative comparison between the redundancies of configurations with different degrees of overlap.

In the fully overlapping case, it is often of interest to evaluate the \emph{asymptotic redundancy}. The asymptotic redundancy of a given class $\mathscr{C}$ of array geometries (with fully overlapping Tx/Rx sensors) is defined as
\begin{align}
	R_\infty(\mathscr{C}) \triangleq \lim_{\substack{|\mathbb{D}|\to\infty\\ \mathbb{D}\in\mathscr{C}}}R(\mathbb{D},\mathbb{D})=\lim_{\substack{|\mathbb{D}|\to\infty\\ \mathbb{D}\in\mathscr{C}}}\frac{|\mathbb{D}|(|\mathbb{D}|+1)}{2(|\mathbb{D}+\mathbb{D}|)}.\label{eq:R_inf}
\end{align}
This quantity is finite only for sparse array configuration achieving a (contiguous) sum co-array of size $|\mathbb{D}+\mathbb{D}|\propto|\mathbb{D}|^2$. Consequently, the asymptotic redundancy of the ULA is infinite, since $|\mathbb{D}+\mathbb{D}|\propto|\mathbb{D}|$. The asymptotic redundancy of any array with a contiguous sum co-array is lower bounded by $R_\infty(\mathscr{C}_\txt{MRA})$, where $\mathscr{C}_\txt{MRA}$ is the class of MRAs, defined in the next subsection.

\subsubsection{Definition of MRA for active sensing}\label{sec:mra_def}

Moffet first proposed the MRA in the context of passive linear arrays and the difference co-array \cite{moffet1968minimumredundancy}. He distinguished between the \emph{general} and \emph{restricted} solutions, where the latter constrains the difference co-array to be contiguous. The restricted MRA thus yields the maximum number of contiguous co-array elements for a given physical aperture. Note that a general MRA can, but need not be restricted---for details see \cite{moffet1968minimumredundancy,rajamaki2021sparsesymmetric} and \cite[pp.~108--109]{rajamaki2021sparsesensor}. Hoctor and Kassam later extended the MRA to active sensing \cite{hoctor1996arrayredundancy}. They specifically considered the restricted MRA with a contiguous sum co-array. Passive and active MRAs correspond to extremal difference bases \cite{redei1949representation,leech1956ontherepresentation,wichmann1963anote} and extremal additive bases \cite{rohrbach1937beitrag,challis1993two,robinson2009some,kohonen2014meet}, respectively. Additive bases have been studied in number theory since the early 1900s and are informally known as solutions to the \emph{postage stamp problem} \cite{lunnon1969postage,mossige1981algorithms}. The canonical version of this problem seeks a set of integers that can represent each non-negative integer smaller than some $ h\in\mathbb{N} $ using the sum of at most $ k $ integers. The restricted postage stamp basis with maximal $ h $ for $ k=2 $, or extremal restricted 2-basis, coincides with the active MRA with fully overlapping Tx and Rx sensors.

For a given number of transmitters $N_{\tx}$, receivers $N_{\rx}$, and transceivers $N$, the MRA can be defined as a (not necessarily unique) solution to \cite{rajamaki2021sparsesensor}
\begin{align}
	\begin{aligned}
		\underset{\mathbb{D}_{\xx}\subset\mathbb{N};\ h\in\mathbb{N}}{\text{maximize}}\ h\ \
		\text{ subject to}\ \ & \mathbb{D}_{\tx}+\mathbb{D}_{\rx}= [0:h-1]\\
		& |\mathbb{D}_{\tx}|=N_{\tx},\  |\mathbb{D}_{\rx}|=N_{\rx}, \txt{ and } |\mathbb{D}_{\tx}\cap \mathbb{D}_{\rx}| = N. 
	\end{aligned}\label{p:MRA_gen}
\end{align}
\cref{p:MRA_gen} subsumes the conventional definition of the (restricted) active MRA in the case of fully overlapping Tx and Rx arrays \cite{hoctor1996arrayredundancy}; If $N_{\tx}=N_{\rx}=N$, then \eqref{p:MRA_gen} reduces to \cite{rajamaki2021sparsesymmetric}
\begin{align}
	\underset{\mathbb{D}\subset \mathbb{N}, h\in\mathbb{N}}{\text{maximize}}\ h\quad
	\text{subject to}\quad \mathbb{D}+\mathbb{D} = [0:h-1],\quad |\mathbb{D}|=N. \label{p:MRA_fo}
\end{align}
The fully overlapping and non-overlapping array configurations in \cref{fig:array_fo,fig:array_no} are examples of MRAs for $N_{\tx}=N_{\rx}=N=4$ and $N_{\tx}=N_{\rx}=N+2=3$, respectively. The fully overlapping MRA is redundant ($R=1.25$), whereas the non-overlapping array is non-redundant ($R=1$).

The MRA can alternatively be defined for a constrained physical aperture \cite{blanton1991newsearch,schwartau2021large} or virtual aperture \cite{rajamaki2021sparsesensor}. In the former case, the MRA is also known as the \emph{sparse ruler} \cite{shakeri2012directionofarrival,romero2016compressive}. This definition of the MRA is convenient when only a limited area for placing the sensors is available, or when the desired physical/virtual array is \emph{not} one-dimensional. For example, constraining the aspect ratio of the planar or volumetric MRA \cite{pumphrey1993design,meurisse2001bounds,kohonen2018planaradditive} precludes degenerate linear arrays as solutions, see \cite[pp.~57--60]{rajamaki2021sparsesensor}.

\subsubsection{Known MRAs} \label{sec:mra_known}
\paragraph{Non-overlapping Tx and Rx sensors}

Problem \eqref{p:MRA_gen} has a simple closed-form solution in the case of non-overlapping Tx and Rx arrays. This non-redundant array, illustrated in \cref{fig:array_no}, has a nested structure consisting of a dense and a sparse ULA. Specifically, the Tx and Rx sensor positions sets are
\begin{align}
		\mathbb{D}_{\tx} = \{0,N_{\rx},\ldots,(N_{\tx}-1)N_{\rx}\}\quad \txt{and}\quad \mathbb{D}_{\rx} = \{0,1,\ldots,N_{\rx}-1\}.\label{eq:MRA_no_nst}
\end{align}
The nested structure in \eqref{eq:MRA_no_nst} is readily verified to be a solution to \eqref{p:MRA_gen}; the upper bound on the size of any contiguous sum co-array $|\mathbb{D}_{\tx}+\mathbb{D}_{\rx}|\leq |\mathbb{D}_{\tx}| |\mathbb{D}_{\rx}|$ is attained since $\mathbb{D}_{\tx}+\mathbb{D}_{\rx}=[0\separator N_{\tx}N_{\rx}-1]$. The definitions of $\mathbb{D}_{\tx}$ and $\mathbb{D}_{\rx}$ may naturally be reversed  in \eqref{eq:MRA_no_nst} without affecting this conclusion.

\cref{eq:MRA_no_nst} is well-known \cite{lockwood1996optimizing1d,li2007mimoradar,cohen2018optimized}, and closely related to the (passive) nested array \cite{pal2010nested}, which also leverages arithmetic sequences to ensure a contiguous (difference) co-array. We note that many non-overlapping MRA configurations exist in the two-dimensional case \cite[pp.~59--60]{rajamaki2021sparsesensor}.

\paragraph{Fully overlapping Tx and Rx sensors}\label{sec:mra_known_fo}

Finding MRAs with fully overlapping Tx and Rx sensors is generally challenging due to the combinatorial nature of \eqref{p:MRA_fo}. Solutions with up to $ |\mathbb{D}|=48$ physical sensors---corresponding to a contiguous sum co-array with $ |\mathbb{D}_\Sigma|=734 $ virtual sensors---have been found using branch-and-bound methods in the context of additive bases \cite{kohonen2015early}. The large search space of the resulting problem is tackled by employing judicious pruning strategies and concatenating solutions of a smaller-sized auxiliary combinatorial problem \cite{kohonen2014meet,kohonen2015early}. \cref{fig:mra_kohonen} illustrates all MRAs with $N=11$ physical sensors (transceivers) \cite{hoctor1996arrayredundancy}. Note that four out of six configurations are symmetric with respect to their mid point (position $11$). In fact, most known MRAs are symmetric; there exist at least one symmetric MRA for each $N\leq 48$ \cite{kohonen2014meet,kohonen2015early}. We emphasize, however, that the MRA in \eqref{p:MRA_fo} is \emph{not} constrained to be symmetric.
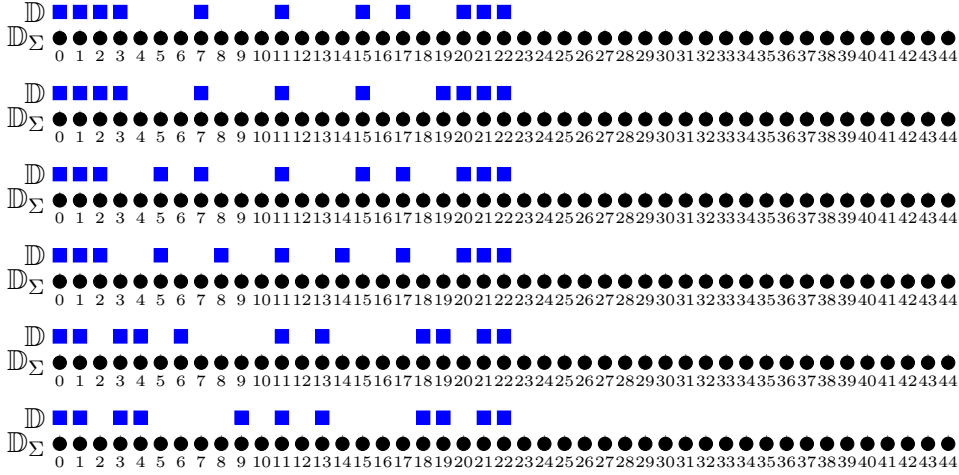
\begin{figure}
 \newcommand{\figwidth}{13.5}
	\newcommand{\Lmra}{22}
	\pgfplotstableread[header=has colnames]{
		x 	y
		0  0
		1 0
		2 0
		3 0
		7 0
		11 0
		15 0
		17 0
		20 0
		21 0
		22 0
	}\MRAa
	\pgfplotstableread[header=has colnames]{
		x	y
		0  0
		1 0
		2 0
		3 0
		7 0
		11 0
		15 0
		19 0
		20 0
		21 0
		22 0
	}\MRAb
	\pgfplotstableread[header=has colnames]{
	x	y
	0  0
	1 0
	2 0
	5 0
	7 0
	11 0
	15 0
	17 0
	20 0
	21 0
	22 0
	}\MRAc
\pgfplotstableread[header=has colnames]{
	x	y
	0  0
	1 0
	2 0
	5 0
	8 0
	11 0
	14 0
	17 0
	20 0
	21 0
	22 0
}\MRAd
\pgfplotstableread[header=has colnames]{
	x	y
	0  0
	1 0
	3 0
	4 0
	6 0
	11 0
	13 0
	18 0
	19 0
	21 0
	22 0
}\MRAe
\pgfplotstableread[header=has colnames]{
	x	y
	0  0
	1 0
	3 0
	4 0
	6 0
	11 0
	13 0
	18 0
	19 0
	21 0
	22 0
}\MRAe
\pgfplotstableread[header=has colnames]{
	x	y
	0  0
	1 0
	3 0
	4 0
	9 0
	11 0
	13 0
	18 0
	19 0
	21 0
	22 0
}\MRAf
\pgfplotsset{every x tick label/.append style={font=\tiny}}
	\centering  
		\begin{tikzpicture} 
			\begin{axis}[width=\figwidth cm,height=2 cm,ytick={-1,0},yticklabels={$\mathbb{D}_\Sigma$,$\mathbb{D}$},xmin=-0.2,xmax=2*\Lmra+0.2,xtick={0,1,...,2*\Lmra},title style={yshift=0 pt},xtick pos=bottom,ytick pos=left,axis line style={draw=none}]
				\addplot[blue,only marks,mark=square*,mark size=2.5] table {\MRAa};
				\addplot[only marks,mark=*,mark size=2.5] expression [domain=0:(2*\Lmra), samples=2*\Lmra+1] {-1};
			\end{axis}
		\end{tikzpicture}
\hfill
		\begin{tikzpicture} 
			\begin{axis}[width= \figwidth cm,height=2 cm,ytick={-1,0},yticklabels={$\mathbb{D}_\Sigma$,$\mathbb{D}$},xmin=-0.2,xmax=2*\Lmra+0.2,xtick={0,1,...,2*\Lmra},title style={yshift=0 pt},xtick pos=bottom,ytick pos=left,axis line style={draw=none}]
				\addplot[blue,only marks,mark=square*,mark size=2.5] table {\MRAb};
					\addplot[only marks,mark=*,mark size=2.5] expression [domain=0:(2*\Lmra), samples=2*\Lmra+1] {-1};
			\end{axis}
		\end{tikzpicture}
\hfill
		\begin{tikzpicture} 
	\begin{axis}[width= \figwidth cm,height=2 cm,ytick={-1,0},yticklabels={$\mathbb{D}_\Sigma$,$\mathbb{D}$},xmin=-0.2,xmax=2*\Lmra+0.2,xtick={0,1,...,2*\Lmra},title style={yshift=0 pt},xtick pos=bottom,ytick pos=left,axis line style={draw=none}]
		\addplot[blue,only marks,mark=square*,mark size=2.5] table {\MRAc};
		\addplot[only marks,mark=*,mark size=2.5] expression [domain=0:(2*\Lmra), samples=2*\Lmra+1] {-1};
	\end{axis}
\end{tikzpicture}
\hfill
		\begin{tikzpicture} 
	\begin{axis}[width= \figwidth cm,height=2 cm,ytick={-1,0},yticklabels={$\mathbb{D}_\Sigma$,$\mathbb{D}$},xmin=-0.2,xmax=2*\Lmra+0.2,xtick={0,1,...,2*\Lmra},title style={yshift=0 pt},xtick pos=bottom,ytick pos=left,axis line style={draw=none}]
		\addplot[blue,only marks,mark=square*,mark size=2.5] table {\MRAd};
		\addplot[only marks,mark=*,mark size=2.5] expression [domain=0:(2*\Lmra), samples=2*\Lmra+1] {-1};
	\end{axis}
\end{tikzpicture}
\hfill
\begin{tikzpicture} 
	\begin{axis}[width= \figwidth cm,height=2 cm,ytick={-1,0},yticklabels={$\mathbb{D}_\Sigma$,$\mathbb{D}$},xmin=-0.2,xmax=2*\Lmra+0.2,xtick={0,1,...,2*\Lmra},title style={yshift=0 pt},xtick pos=bottom,ytick pos=left,axis line style={draw=none}]
		\addplot[blue,only marks,mark=square*,mark size=2.5] table {\MRAe};
		\addplot[only marks,mark=*,mark size=2.5] expression [domain=0:(2*\Lmra), samples=2*\Lmra+1] {-1};
	\end{axis}
\end{tikzpicture}
\hfill
\begin{tikzpicture} 
	\begin{axis}[width= \figwidth cm,height=2 cm,ytick={-1,0},yticklabels={$\mathbb{D}_\Sigma$,$\mathbb{D}$},xmin=-0.2,xmax=2*\Lmra+0.2,xtick={0,1,...,2*\Lmra},title style={yshift=0 pt},xtick pos=bottom,ytick pos=left,axis line style={draw=none}]
		\addplot[blue,only marks,mark=square*,mark size=2.5] table {\MRAf};
		\addplot[only marks,mark=*,mark size=2.5] expression [domain=0:(2*\Lmra), samples=2*\Lmra+1] {-1};
	\end{axis}
\end{tikzpicture}
	\caption{Active MRAs with $N=11$ physical (transceiving) sensors \cite{hoctor1996arrayredundancy}. Most of the configurations are symmetric. In fact, there exists a symmetric MRA configuration for (at least) each $N\leq48$ \cite{kohonen2014meet,kohonen2015early}.}\label{fig:mra_kohonen}
\end{figure}

\subsection{Symmetric sparse array configurations}\label{sec:scalable}

The main drawback of the MRA is that finding it generally requires solving a combinatorial optimization problem. Consequently, generating MRAs with a large number of sensors or large aperture is usually impractical---especially when the Tx and Rx arrays are not non-overlapping. In contrast, if a sparse array configuration allows a parametric description, then its synthesis is simplified considerably. There exist \emph{scalable} and \emph{order-wise optimal} array geometries that can be generated for practically any number of sensors or aperture size while achieving a co-array with on the order of $N_{\tx}N_{\rx}$ contiguous virtual sensors. This allows generating low-redundancy sparse array configurations beyond hundreds (thousands) of physical (virtual) sensors. Such geometries also upper bound the redundancy of the MRA. Together with lower bounds, the potential suboptimality with respect to the MRA can also be estimated; see \cite{rajamaki2021sparsesymmetric}.

This section focuses on a class of \emph{symmetric} sparse array geometries that achieve a low redundancy and contiguous sum co-array. We consider the case of fully overlapping Tx and Rx arrays. Imposing symmetry on the array geometry is not only a simplification inspired by the symmetry of most known MRAs (which are optimal but computationally intensive to find), but it also ensures equivalence between the sum and difference co-arrays. Consequently, properly designed symmetric arrays are well-suited for both active and passive sensing.

\subsubsection{Generic symmetric array} \label{sec:symm_gen}

Let $\mathbb{S}$ denote the set of normalized sensor positions of an array that is symmetric about $ \tfrac{1}{2}\max\mathbb{S}$, i.e., $\mathbb{S} = \max\mathbb{S}-\mathbb{S}$. Here, $\max\mathbb{S}$ is the largest element of $\{0\}\subseteq\mathbb{S}\subset \mathbb{N}$.\footnote{Consequently, position $ \tfrac{1}{2}\max\mathbb{S}$ contains a sensor if and only if $|\mathbb{S}|$ is odd.} This symmetry implies that the sum and difference co-array of $\mathbb{S}$ are identical up to a shift \cite[p.~740]{hoctor1990theunifying}:
\begin{align}
	\mathbb{S} = \max\mathbb{S}-\mathbb{S} \implies  \mathbb{S}+ \mathbb{S}=\mathbb{S}- \mathbb{S}+\max\mathbb{S}.\label{eq:symm_sum_diff}
\end{align}
Symmetric arrays are thus suited for both passive and active sensing. Symmetry also greatly simplifies array design by enabling the synthesis of arrays with a contiguous sum co-array using configurations with a contiguous difference co-array, such as the well-known Nested array \cite{pal2010nested}.

\paragraph{Structure of symmetric array}
Any symmetric (linear) array can be described by a generator array $\{0\}\subseteq\mathbb{G}\subset\mathbb{N}$ and a non-negative offset $ \ell\in\mathbb{N} $, which determines the spacing between $ \mathbb{G} $ and its mirror image $ \max\mathbb{G}-\mathbb{G} $. Formally, the \emph{symmetric array} $\mathbb{S}=\mathbb{S}(\mathbb{G},\ell)$ can be defined as\footnote{This definition is general, since symmetric arrays generated by $\mathbb{G}$ and \emph{nonpositive} offset $\ell$ can equivalently be generated by $\mathbb{G}' =\max\mathbb{G}-\mathbb{G}$ and nonnegative $\ell'$.}
\begin{align}
		\mathbb{S}(\mathbb{G},\ell) \triangleq\mathbb{G}\cup (\max\mathbb{G}-\mathbb{G}+\ell).\label{eq:symm_gen}
\end{align}
We are specifically interested in $(\mathbb{G},\ell)$ pairs yielding an $\mathbb{S}$ with a contiguous sum co-array. The desired property $\mathbb{S}+\mathbb{S}=[0:2\max\mathbb{S}]$ can naturally be decomposed using \eqref{eq:symm_gen} into necessary and sufficient conditions involving offset $\ell$ and the sum/difference co-array of generator $\mathbb{G}$ \cite[Theorem~2]{rajamaki2021sparsesymmetric}. However, a simple sufficient condition stands out as particularly useful: If $ \mathbb{G} $ has a contiguous difference co-array and $ \ell $ is no larger than the number elements in the hole-free segment of the sum co-array of $\mathbb{G}$ (including $\{0\}$), then $\mathbb{S}$ has a contiguous sum and difference co-array (cf. \cite[Theorem~2 and Corollary~2]{rajamaki2021sparsesymmetric}).
\begin{theorem}[Sufficient condition for contiguous co-array]\label{thm:symm_gen}
	Consider the symmetric array $\mathbb{S}$ in \eqref{eq:symm_gen} with generator $\mathbb{G}$ and offset $\ell$. If $ \mathbb{G}-\mathbb{G}\supseteq [0\separator \max\mathbb{G}]$ and $  \mathbb{G}+\mathbb{G}\supseteq[0\separator \ell-1] $, then $\mathbb{S}+\mathbb{S}=[0\separator 2\max\mathbb{S}]$. 
\end{theorem}
\begin{proof}
 	By \eqref{eq:symm_gen}, the sum co-array of the symmetric array $\mathbb{S}$ is
	\begin{align*}
		\mathbb{S}+\mathbb{S}&=\mathbb{G}\cup(\max \mathbb{G}-\mathbb{G}+\ell)+ \mathbb{G}\cup(\max \mathbb{G}-\mathbb{G}+\ell)\\
		&=(\mathbb{G}+\mathbb{G})\cup(\mathbb{G}-\mathbb{G}+\max \mathbb{G}+\ell)\cup (2(\max\mathbb{G}+\ell)-(\mathbb{G}+\mathbb{G})).
		\end{align*}
		By the symmetry of the difference co-array, $ \mathbb{G}-\mathbb{G}\supseteq [0\separator \max\mathbb{G}]\implies \mathbb{G}-\mathbb{G}\supseteq [-\max\mathbb{G}\separator \max\mathbb{G}]$. Combining this with $  \mathbb{G}+\mathbb{G}\supseteq[0\separator \ell-1] $ implies that
		\begin{align*}
		\mathbb{S}+\mathbb{S}&\supseteq [0\separator \ell-1]\cup([-\max\mathbb{G}\separator \max\mathbb{G}]+\max\mathbb{G}+\ell)\cup(2(\max\mathbb{G}+\ell)-[0\separator \ell-1])\\
		&=[0\separator \ell-1]\cup[\ell \separator 2\max\mathbb{G}+\ell]\cup[2\max\mathbb{G}+\ell+1\separator 2(\max\mathbb{G}+\ell)]\\
		&=[0\separator 2(\max\mathbb{G}+\ell)].
	\end{align*}
	Noting that $\max\mathbb{G}+\ell=\max\mathbb{S}$ by \eqref{eq:symm_gen} completes the proof.\hfill$\blacksquare$
\end{proof}

A key takeaway of \cref{thm:symm_gen} is that the design of (symmetric) array configurations with a contiguous sum co-array can be considerably simplified by leveraging geometries with a contiguous difference co-array.

\paragraph{Minimum-redundancy symmetric array over generator class $\mathscr{G}$}
If $\mathbb{G}$ admits a simple parametric description, both the parameters of $\mathbb{G}$ and the offset $\ell$ may be optimized for minimal redundancy while imposing the sufficient conditions of \cref{thm:symm_gen} as constraints. Let $\mathscr{G}$ denote any class of array geometries with a contiguous difference co-array, i.e., $ \mathbb{G}-\mathbb{G}\supseteq [0\separator \max\mathbb{G}]\ \forall \mathbb{G}\in\mathscr{G}$. For a given generator class $\mathscr{G}$ and number of physical sensors $N$, the minimum-redundancy symmetric array over $\mathscr{G}$ is defined as any solution to
\begin{align}
	\underset{\mathbb{G}\in\mathscr{G}, \ell \in \mathbb{N}}{\txt{maximize}}\ \max \mathbb{G}+\ell\quad 
	\text{subject to}\quad 
	\begin{cases}
		|\mathbb{G}\cup (\max\mathbb{G}-\mathbb{G}+\ell)|=N\\
		\mathbb{G}+\mathbb{G}\supseteq [0\separator \ell-1].
	\end{cases}
	 \label{p:SA}
\end{align}
In particular, if \eqref{p:SA} has a (not necessarily unique) solution $(\mathbb{G}^\star,\ell^\star)$ for a given $\mathscr{G}$ and $N$, then we denote by  $\mathbb{S}(\mathbb{G}^\star,\ell^\star)$ a minimum-redundancy symmetric array for generator class $\mathscr{G}$. For general $\mathscr{G}$, \eqref{p:SA} is a combinatorial problem which may be impractical to solve. However, certain choices of $\mathscr{G}$ actually admit a closed-form solution for any $N\in\mathbb{N}_+$, as we will see in the next subsection.

\subsubsection{Symmetric nested array}

We now illustrate the general symmetric array design framework developed in \cref{sec:symm_gen}. As an example of a simple yet versatile generator, we use the Nested array (NA) \cite{pal2010nested}, which has a contiguous difference co-array. By \cref{thm:symm_gen}, the NA generator guarantees that the resulting symmetric array $\mathbb{S}$ has a contiguous sum co-array for suitable choices of offset $\ell$. The convenient parametric structure of the symmetric NA also enables finding its redundancy minimizing parameters in a closed form. The symmetric NA and the NA generator are illustrated in \cref{fig:NA_symm}.
\begin{figure}[]
	\centering
	\subfloat[Nested Array (NA)]{\includegraphics[width=.6\linewidth]{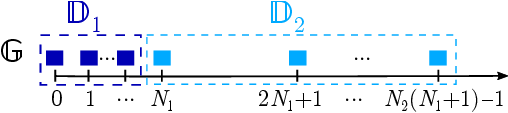}\label{fig:G_NA}}\\
	\subfloat[Symmetric NA with offset $\ell>N_1$ ]{\includegraphics[width=.9\linewidth]{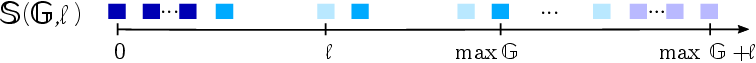}\label{fig:S-NA}}\\
	\subfloat[Symmetric NA with offset $\ell=N_1$ ]{\includegraphics[width=.8\linewidth]{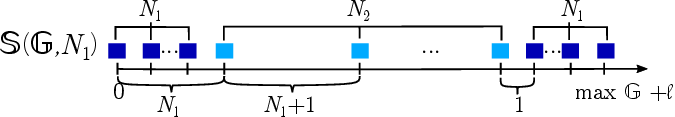}\label{fig:CNA}}
	\caption{The \protect\subref{fig:G_NA} NA generator $\mathbb{G}=\mathbb{D}_\txt{NA}$ and shift parameter $\ell$ define the \protect\subref{fig:S-NA} symmetric NA, which reduces to the configuration in \protect\subref{fig:CNA}, when $\ell=N_1$. The minimum-redundancy symmetric NA solving \eqref{p:SA} has the structure in \protect\subref{fig:CNA}. Adapted from \cite{rajamaki2021sparsesymmetric} {\copyright} 2021 IEEE}\label{fig:NA_symm}
\end{figure}

Let $\mathscr{G}$ be the class of NAs. In this case, the generator is $\mathbb{G}=\mathbb{D}_\txt{NA}=\mathbb{D}_\txt{NA}(N_1,N_2)$ with nonnegative parameters $ N_1,N_2 \in\mathbb{N}$ and
\begin{align}
	\mathbb{D}_\txt{NA}\!\triangleq\!\mathbb{D}_1\!\cup\!(\mathbb{D}_2+N_1), \txt{ where }
	\begin{cases}
		\mathbb{D}_1\!=\![0\!\separator\!N_1-1]\\
		\mathbb{D}_2\!=\![0\!\separator\!N_1+1\!\separator (N_2-1)(N_1+1)].
	\end{cases}\label{eq:G_s-na}
\end{align}
By \eqref{eq:symm_gen}, the \emph{symmetric} NA can therefore be written as
\begin{align}
	\mathbb{S}_\txt{NA}(N_1,N_2,\ell)
	\triangleq \mathbb{S}(\mathbb{D}_\txt{NA},\ell)=\mathbb{D}_\txt{NA}\cup  (\max \mathbb{D}_\txt{NA} - \mathbb{D}_\txt{NA}+\ell).\label{eq:s-na}
\end{align}
Triplet $(N_1,N_2,\ell)$ controls the redundancy of $\mathbb{S}_\txt{NA}$. For example, setting $ \ell=0 $ can increase redundancy and consequently improve robustness to sensor failures \cite[Definition~2]{liu2018optimizing}. Of special interest herein is the case $\ell = N_1$, which yields the following low-redundancy geometry:\footnote{This case is equivalent to $\ell =(N_1+1)k+N_1$ for any $k\in[0\separator N_2]$ by reparametrization.}
\begin{align}
	\begin{aligned}
	\mathbb{D}_\txt{CNA}(N_1,N_2)&\triangleq 
	\mathbb{S}_\txt{NA}(N_1,N_2,\ell=N_1)\\
	&=
	\mathbb{D}_1 \cup (\mathbb{D}_2 +N_1) \cup (\mathbb{D}_1+(N_1+1)N_2).\label{eq:cna}
\end{aligned}
\end{align}
This array configuration, originally studied in the context of additive combinatorics \cite[Satz~2]{rohrbach1937beitrag}, is also known as the Concatenated nested array (CNA) \cite{rajamaki2017sparselinear}. The solution to \eqref{p:SA} actually has the same form as \eqref{eq:cna} when the elements of $\mathscr{G}$ are given by \eqref{eq:G_s-na} \cite[Proposition~1]{rajamaki2021sparsesymmetric}. Intuitively, this is because any other offset $\ell $ cannot decrease redundancy. Due to its simple structure, the minimum-redundancy parameters of \eqref{eq:cna}, and by implication \eqref{eq:s-na}, can be found in a closed form. For any given $ N=4m+k$, where $ m\in \mathbb{N}, k \in \{0:3\}$, and the elements of $\mathscr{G}$ follow \eqref{eq:G_s-na}, the solution to \eqref{p:SA} is $\mathbb{G}^\star=\mathbb{D}_\txt{CNA}(N_1^\star,N_2^\star)$, $\ell^\star=N_1^\star$, where \cite[Theorem~3]{rajamaki2021sparsesymmetric}
\begin{align*}
	N_1^\star = (N-\alpha)/4,\quad
	N_2^\star = (N+\alpha)/2,\quad \txt{and}\quad \alpha = (k+1)\bmod 4 -1.
\end{align*}
We note that the optimal values of $N_1$ and $N_2$ are similar for the minimally redundant NA (w.r.t. the difference co-array), where $N_1^\star= N_2^\star= N/2$ for even $N$ \cite{pal2010nested}. However, while the NA has a contiguous difference co-array, its sum co-array is non-contiguous. Indeed, \cref{fig:na_symm_ex} shows that the symmetric NA achieves a larger contiguous sum co-array than the NA for the same physical aperture or number of physical sensors.
\begin{figure}
	\newcommand{\figwidth}{13}
	\newcommand{\figheight}{3}
	\newcommand{\Lmax}{29}
	\newcommand{\offs}{19}
	\pgfplotsset{every x tick label/.append style={font=\tiny},every y tick label/.append style={font=\scriptsize}}
	\centering
	\subfloat[Nested array, $N=8$ sensors]{    
		\begin{tikzpicture}
			\begin{axis}[width=\figwidth cm ,height=\figheight cm,xmin=-0.2,xmax=2*\Lmax+0.2,xtick={0,2,...,2*\Lmax},ytick={-1,0,1},yticklabels={$\mathbb{D}_\Delta+\offs$,$\mathbb{D}_\Sigma$,$\mathbb{D}$},ytick style={draw=none},ymin=-1.5,ymax=1.5,xticklabel shift = 0 pt,xtick pos=bottom,axis line style={draw=none}]
				\addplot[blue,only marks,mark=square*,mark size=2,y filter/.code={\pgfmathparse{\pgfmathresult*0+1}}] table[x=d,y=d]{Data/D_NA.dat};
				\addplot[only marks,mark=*,mark size=2.5,y filter/.code={\pgfmathparse{\pgfmathresult*0}}] table[x=d,y=d]{Data/D_Sigma_NA.dat};
				\addplot[gray,only marks,mark=*,mark size=2.5,y filter/.code={\pgfmathparse{\pgfmathresult*0-1}},x filter/.code={\pgfmathparse{\pgfmathresult+\offs}}] table[x=d,y=d]{Data/D_Delta_NA.dat};
			\end{axis}
		\end{tikzpicture}\label{fig:na_L}
	}\\%
	\subfloat[Nested array, $N=10$ sensors]{    
		\begin{tikzpicture}
			\begin{axis}[width=\figwidth cm ,height=\figheight cm,xmin=-0.2,xmax=2*\Lmax+0.2,xtick={0,2,...,2*\Lmax},ytick={-1,0,1},yticklabels={$\mathbb{D}_\Delta+\Lmax$,$\mathbb{D}_\Sigma$,$\mathbb{D}$},ytick style={draw=none},ymin=-1.5,ymax=1.5,xticklabel shift = 0 pt,xtick pos=bottom,axis line style={draw=none}]
				\addplot[blue,only marks,mark=square*,mark size=2,y filter/.code={\pgfmathparse{\pgfmathresult*0+1}}] table[x=d,y=d]{Data/D_NA_alt.dat};
				\addplot[only marks,mark=*,mark size=2.5,y filter/.code={\pgfmathparse{\pgfmathresult*0}}] table[x=d,y=d]{Data/D_Sigma_NA_alt.dat};
				\addplot[gray,only marks,mark=*,mark size=2.5,y filter/.code={\pgfmathparse{\pgfmathresult*0-1}},x filter/.code={\pgfmathparse{\pgfmathresult+\Lmax}}] table[x=d,y=d]{Data/D_Delta_NA_alt.dat};
			\end{axis}
		\end{tikzpicture}\label{fig:na_N}
	}\\
	\subfloat[Symmetric Nested array, $N=10$ sensors]{    
		\begin{tikzpicture}
			\begin{axis}[width=\figwidth cm ,height=\figheight cm,xmin=-0.2,xmax=2*\Lmax+0.2,xtick={0,2,...,2*\Lmax},ytick={-1,0,1},yticklabels={$\mathbb{D}_\Delta+\offs$,$\mathbb{D}_\Sigma$,$\mathbb{D}$},ytick style={draw=none},ymin=-1.5,ymax=1.5,xticklabel shift = 0 pt,xtick pos=bottom,axis line style={draw=none}]
				\addplot[blue,only marks,mark=square*,mark size=2,y filter/.code={\pgfmathparse{\pgfmathresult*0+1}}] table[x=d,y=d]{Data/D_CNA.dat};
				\addplot[only marks,mark=*,mark size=2.5,y filter/.code={\pgfmathparse{\pgfmathresult*0}}] table[x=d,y=d]{Data/D_Sigma_CNA.dat};
				\addplot[gray,only marks,mark=*,mark size=2.5,y filter/.code={\pgfmathparse{\pgfmathresult*0-1}},x filter/.code={\pgfmathparse{\pgfmathresult+\offs}}] table[x=d,y=d]{Data/D_Delta_CNA.dat};
			\end{axis}
		\end{tikzpicture}\label{fig:sna_N}
	}
	\caption{NA and symmetric NA configurations using minimum-redundancy parameters, with corresponding sum co-array $\mathbb{D}_\Sigma$ and difference co-array $\mathbb{D}_\Delta$ (shifted to non-negative half plane). The \protect\subref{fig:sna_N} symmetric NA attains a larger contiguous sum co-array than the NA, both for the same \protect\subref{fig:na_L} physical array aperture and \protect\subref{fig:na_N} number of physical sensors. The difference co-array of the symmetric NA is also contiguous.}\label{fig:na_symm_ex}
\end{figure}
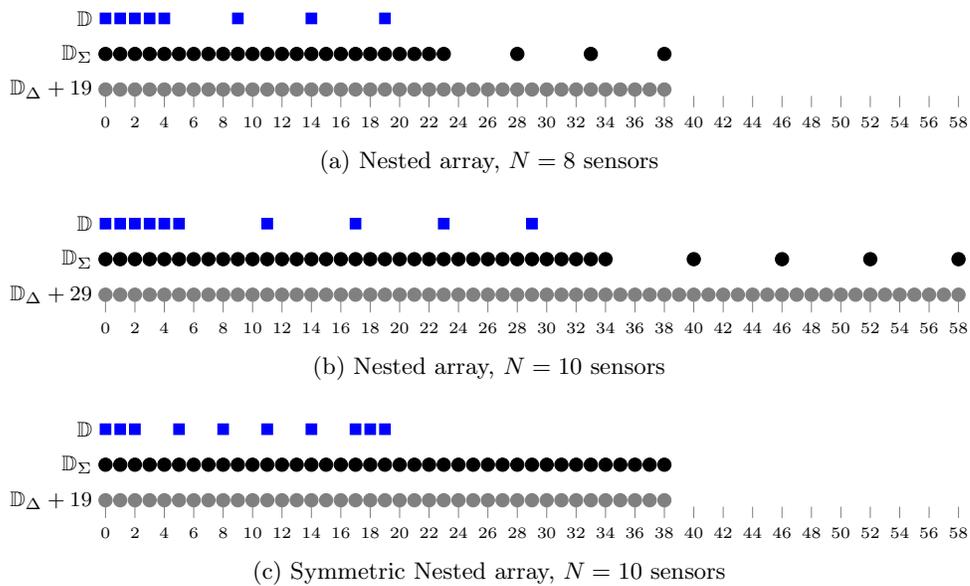

 The minimum-redundancy symmetric NA can be used to upper bound the asymptotic redundancy of the MRA. Specifically, denote the class of minimum-redundancy symmetric NAs by $\mathscr{C}^\star_\txt{S-NA}\triangleq\{\mathbb{S}_\txt{NA}^\star(N)\ |\ N\in\mathbbm{N}_+\} $, where $\mathbb{S}_\txt{NA}^\star(N)\triangleq \mathbb{D}_\txt{CNA}(N_1^\star,N_2^\star)$. The number of contiguous sum co-array elements can be shown to be $|\mathbb{S}_\txt{NA}^\star+\mathbb{S}_\txt{NA}^\star|=2\max\mathbb{S}_\txt{NA}^\star+1= (N^2+6N-7)/4-(\alpha -1)^2/4+1$. By \eqref{eq:R_inf}, the asymptotic redundancy of $\mathscr{C}^\star_\txt{S-NA}$ is therefore
\begin{align}
	R_\infty(\mathscr{C}^\star_\txt{S-NA})=\lim_{N\to\infty}\frac{N(N+1)}{2((N^2+6N-7)/4-(\alpha -1)^2/4+1)}=2. \label{eq:R_inf_s-na}
\end{align}
Consequently, the asymptotic redundancy of the MRA satisfies $R_\infty(\mathscr{C}_\txt{MRA})\leq R_\infty(\mathscr{C}_\txt{S-NA}^\star)=2$, where $\mathscr{C}_\txt{MRA}$ denotes the set of solutions to \eqref{p:MRA_fo}.

\subsubsection{Other symmetric arrays}

Beyond the NA, previously studied generator choices include the Wichmann \cite{wichmann1963anote} and Kl{\o}ve-Mossige bases \cite{mossige1981algorithms} originating in number theory. For array geometry design perspectives and technical details, see \cite{rajamaki2021sparsesymmetric}. See also \cite{liu2017maximally,cohen2020sparse} for symmetric arrays with a fractal structure. We remark that alternative generators in \eqref{eq:symm_gen} can improve upon the redundancy of the symmetric NA. In particular, a subclass $\mathcal{K}$ of the symmetric array generated by the Kl{\o}ve-Mossige basis---corresponding to a class of symmetric bases studied in the context of additive combinatorics \cite{klove1981class}---achieves asymptotic redundancy $R_\infty(\mathcal{K}) = 23/12 < 1.9167$ \cite[Theorem, p.~177]{klove1981class} (cf. \cite[Proposition~2]{rajamaki2021sparsesymmetric}). This lower bounds \eqref{eq:R_inf_s-na}, and therefore provides a tighter upper bound on the asymptotic redundancy of the (fully overlapping) MRA. Array configurations---symmetric or asymmetric---with an even lower asymptotic redundancy, yet a contiguous sum co-array, remain to be explored.

\section{Beamforming}\label{sec:coarray_beamforming}

This section considers \emph{beamforming} in sensing using sparse arrays. Beamforming collectively denotes spatially directive transmission or spatial filtering upon reception, which is used to enhance (suppress) signals impinging in/from directions of interest by the principle of constructive (destructive) wave interference.

\cref{sec:overview_bf} starts by briefly reviewing Tx, Rx, and joint Tx-Rx beamforming. We focus on the joint case, which is commonly encountered in imaging applications. \cref{sec:image_addition} considers an extension to Tx-Rx beamforming called \emph{image addition}, where carefully chosen Tx and Rx beamforming weight pairs are used to synthesize a desired Tx-Rx beampattern by linearly combining multiple ``component images''. This can be interpreted as assigning desired complex values to the beamforming weights of the sum co-array. We briefly discuss methods for finding the physical beamforming weights synthesizing a desired virtual array weighting, and describe the resulting trade-off between array sparsity and the number of component images, $Q$. A small value of $Q$ is often preferred to minimize the beampattern synthesis time.

\subsection{Rx, Tx, and joint Tx-Rx beamforming}\label{sec:overview_bf}
In transmit beamforming, the spatial self-interference pattern of the radiated wavefield, or so-called transmit beampattern, is controlled by the cross-correlation properties of the transmitted waveforms. In receive beamforming, the outputs of the Rx sensors are scaled and phase shifted prior to being coherently combined. The receive beampattern can be flexibly shaped and steered post data-acquisition when the beamforming architecture of the receiver is fully digital. In contrast, the transmit beampattern can typically change only between transmissions, where a transmission corresponds to the duration of the waveform or the maximum round-trip time that should be unambiguously resolved. The \emph{effective} Tx-Rx beampattern equals the product of the Tx and Rx beampatters, and hence achieves a narrower main beam or lower side lobe levels than the individual Tx and Rx beampatterns, as \cref{fig:beampattern} illustrates. A properly designed Tx-Rx beampattern also suppresses grating lobes present in either the Tx or Rx beampattern due to spatial aliasing.
\begin{figure}
	\centering
	\includegraphics[width=0.7\textwidth]{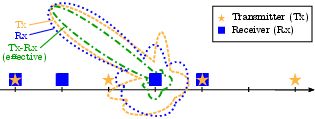}
	\caption{The effective Tx-Rx beampattern of an active array equals the product of the Tx and Rx beampatterns. Hence, it typically achieves a narrower main lobe and lower side lobe levels than either the Tx or Rx beampattern, which may suffer from grating lobes in case the physical array geometry is sparse. Adapted from \cite{rajamaki2020hybrid} {\copyright} 2020 IEEE.}\label{fig:beampattern}
\end{figure}

\subsubsection{Receive beamforming} 
The Rx beampattern magnitude, i.e., the angular magnitude response of the Rx array employing beamforming weight vector $\bm{w}_{\rx}\in\mathbb{C}^{N_{\rx}}$ is
\begin{align*}
	B_{\rx}(\theta)= |H_{\rx}(\theta)|^2= |\bm{w}_{\rx}^{\HT}\bm{a}_{\rx}(\theta)|^2,
\end{align*}
where $H_{\rx}(\theta)\in\mathbb{C}$ is the \emph{complex-valued Rx beampattern} or angular response,
\begin{align}
	H_{\rx}(\theta)\triangleq \bm{w}_{\rx}^{\HT}\bm{a}_{\rx}(\theta).\label{eq:H_r}
\end{align}
A typical choice for the Rx beamforming vector $\bm{w}_{\rx}$ is the spatial matched filter $\bm{w}_{\rx}=c \bm{a}_{\rx}(\phi)$, where $c\in\mathbb{C}\setminus\{0\}$. The matched filter maximizes the SNR of the beamformer output in case of a single scatterer in direction $\phi$ embedded in zero-mean white noise (as a simple consequence of the Cauchy-Schwartz inequality). Other choices of $\bm{w}_{\rx}$ can strike a different balance between main lobe width and side lobe levels, or adapt to the scattering environment
\cite{vanveen1988beamforming}. The widely used MVDR (minimum variance distortionless response) beamformer is a prime example of such a design.
 
The beamformed output of the Rx array gives rise to a multiple-input single-output (MISO) model. Specifically, the received signal matrix $\bm{Y}$ in \eqref{eq:Y} after Rx beamforming reduces to a $T$-dimensional \emph{temporal} (row) vector
\begin{align*}
	\bm{w}_{\rx}^{\HT}\bm{Y}
	=\sum_{k=1}^K\gamma_k H_{\rx}(\theta_k)\bm{a}_{\tx}^{\T}(\theta_k)\bm{S}^{\T}+\bm{w}_{\rx}^{\HT}\bm{N}.
\end{align*}
Rx beamforming is a key component in wireless communications, radar and passive sensing-based direction finding systems, in part due to its low computational complexity (once an appropriate $\bm{w}_{\rx}$ is found). In practice, Rx beamforming may also be employed out of necessity due to hardware constraints imposed by a hybrid or an analog beamforming architecture. 

\subsubsection{Transmit beamforming}

The angular power distribution of the wavefield radiated by the Tx array, or Tx beampattern magnitude, is given by \cite{fuhrmann2008transmit}
\begin{align*}
	B_{\tx}(\theta)=	\bm{a}_{\tx}^{\HT}(\theta)\bm{S}^{\HT}\bm{S}\bm{a}_{\tx}(\theta).
\end{align*}
Clearly, $B_{\tx}(\theta)$ depends on the cross-correlation properties of the transmitted waveforms via positive semidefinite matrix $\bm{R}_{\tx}\triangleq \bm{S}^{\HT}\bm{S}$. For example, $\bm{R}_{\tx}=\bm{a}_{\tx}(\phi)\bm{a}_{\tx}^{\HT}(\phi)$ and $\bm{R}_{\tx}=\bm{I}$ correspond to highly directive and omnidirectional transmission, respectively. The former case is also known as the phased array or spatial matched Tx filter (in direction $\phi$), whereas the latter is typically referred to as the orthogonal MIMO model. \cref{fig:Bt_examples} illustrates these two cases for a ULA and an MRA transmitter with $N_{\tx}=9$ sensors. The sparse array achieves a narrower main lobe due to its larger aperture, but its sidelobe levels are higher due to the spatial undersampling. Note that the waveform rank, $N_s=\rank(\bm{S})=\rank(\bm{R}_{\tx})$, differs in the phased array ($N_s=1$) and orthogonal MIMO ($N_s=N_{\tx}$) cases. Increasing $N_s$ generally extends the set of achievable Tx beampatterns---for example, see \cite{stoica2007onprobing}.%
\begin{figure}
	\centering
     \begin{tikzpicture}
	              \begin{axis}[
	              	width=6.3 cm,height=.0cm,
		              hide axis,
		              scale only axis,
		              legend style = {draw=none,fill=none},legend columns=2
		               ]
		                  \addplot[black,very thick,draw=none] table[x=phi,y=B]{Data/Bt_1.dat};
		                  \addlegendentry{phased array}
		                  \addplot[gray,very thick, dashed,draw=none] table[x=phi,y=B]{Data/Bt_2.dat};
		                  \addlegendentry{orthogonal MIMO}
		              \end{axis}
	          \end{tikzpicture}\vspace{-.3cm}\\
	\subfloat[ULA ($N_{\tx}=9$)]{
	\begin{tikzpicture}
		\begin{axis}[width=6.3 cm,height=4 cm,ymin=0,ylabel={$B_{\tx}(\theta)$},xmin=-pi/2,xmax=pi/2,xlabel={Angle, $\theta$},xtick={-pi/2,0,pi/2},xticklabels={$-\pi/2$,$0$,$\pi/2$},title style={yshift=0 pt},xticklabel shift = 0 pt,xlabel shift = {-5 pt},yticklabel shift=0pt,ylabel shift = 0 pt,ymode=log,ymin=0.001,ymax=1.2]
			\addplot[black,very thick] table[x=phi,y=B]{Data/Bt_1.dat};
			\addplot[gray,very thick, dashed] table[x=phi,y=B]{Data/Bt_2.dat};
		\end{axis}
	\end{tikzpicture}
}
\subfloat[MRA ($N_{\tx}=9$)]{
\begin{tikzpicture}
	\begin{axis}[width=6 cm,height=4cm,ymin=0,xmin=-pi/2,xmax=pi/2,xlabel={Angle, $\theta$},xtick={-pi/2,0,pi/2},xticklabels={$-\pi/2$,$0$,$\pi/2$},title style={yshift=0 pt},xticklabel shift = 0 pt,xlabel shift = {-5 pt},yticklabel shift=0pt,ylabel shift = 0 pt,ymode=log,ymin=0.001,ymax=1.2]
		\addplot[black,very thick] table[x=phi,y=B]{Data/Bt_1_sparse.dat};
		\addplot[gray,very thick, dashed] table[x=phi,y=B]{Data/Bt_2_sparse.dat};
	\end{axis}
\end{tikzpicture}
}
	\caption{Transmit beampattern (magnitude squared) of phased array (steered to boresight) and orthogonal MIMO transmission. For the same Tx power, the phased array yields higher beamforming gain in the look direction, whereas orthogonal waveforms illuminate a wider field-of-view.}
	\label{fig:Bt_examples}
\end{figure}
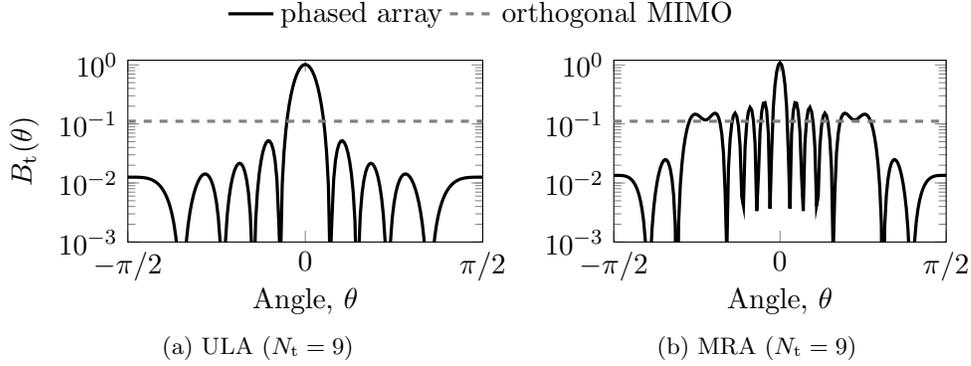
 
For brevity, the remainder of the chapter focuses on the case of unit waveform rank $N_s=1$, i.e. the same waveform is launched from each sensor with an appropriate phase shift. In this case, the \emph{complex-valued Tx beampattern} $H_{\tx}(\theta)\in\mathbb{C}$ can be defined analogously to the Rx case in \eqref{eq:H_t} as
 \begin{align}
 	H_{\tx}(\theta)\triangleq \bm{w}_{\tx}^{\HT}\bm{a}_{\tx}(\theta),
 	\label{eq:H_t}
 \end{align}
 where $\bm{w}_{\tx} \triangleq \bm{S}^{\HT}\bm{w}_s\in\mathbb{C}^{N_{\tx}}$ is a Tx beamforming weight vector and $\bm{w}_s\in\mathbb{C}^T$ is the corresponding temporal waveform combining vector. Since $N_s=1$, $|H_{\tx}(\theta)|^2\propto B_{\tx}(\theta)$ holds for any $\bm{w}_s$ not in the null space of $\bm{S}^{\HT}$. To see this, write $\bm{S}=\bm{u}\bm{v}^{\HT}$, where $\bm{u}\in\mathbb{C}^{T}$, $\bm{v}\in\mathbb{C}^{N_{\tx}}$ are nonzero vectors and $\|\bm{u}\|_2=1$. Consequently, $|H_{\tx}(\theta)|^2=|\bm{w}_s^{\HT}\bm{u}|^2|\bm{v}^{\HT}\bm{a}_{\tx}(\theta)|^2=|\bm{w}_s^{\HT}\bm{u}|^2 B_{\tx}(\theta)$. A reasonable choice for $\bm{w}_s$ is the temporal matched filter $\bm{w}_s=c \bm{u}$, for some $c\in\mathbb{C}\setminus\{0\}$. We remark that if $N_s\geq 2$, the interpretation of \eqref{eq:H_t} changes. In this case, $H_{\tx}(\theta)$ is a partly \emph{virtual} Tx beampattern defined by both waveform matrix $\bm{S}$, which affects the radiated wavefield, and combining vector $\bm{w}_s$, which does not. For example, in the case of orthogonal waveforms (implying full waveform rank, $N_s=N_{\tx}$), the transmit power is radiated omnidirectionally. Hence all beamforming is performed entirely synthetically post-reception.

The received signal after Tx beamforming constitutes a SIMO model, which takes the form of an $N_{\rx}$-dimensional \emph{spatial} vector
\begin{align*} 
	\bm{Y}\bm{w}_s^\ast
	=\sum_{k=1}^K\gamma_k\bm{a}_{\rx}(\theta_k)H_{\tx}(\theta_k)+\bm{N}\bm{w}_s^\ast.
\end{align*}

\subsubsection{Joint transmit-receive beamforming}\label{sec:joint_txrx_beamforming}
By \labelcref{eq:H_r,eq:H_t}, the \emph{complex-valued Tx-Rx beampattern} is defined as
\begin{align}
	H_{\txrx}(\theta)\triangleq H_{\tx}(\theta)H_{\rx}(\theta)
	=(\bm{w}_{\tx}\kron \bm{w}_{\rx})^{\HT}( \bm{a}_{\tx}(\theta)\kron \bm{a}_{\rx}(\theta)).
	\label{eq:H_tr}
\end{align}
By \eqref{eq:Y}, the output after joint Tx-Rx beamforming can thus be written as
\begin{align}
	y\triangleq	\bm{w}_{\rx}^{\HT}\bm{Y}\bm{w}_s^\ast 
	=\sum_{k=1}^K \gamma_kH_{\txrx}(\theta_k) +\tilde{n},
	\label{eq:y_txrx}
\end{align}
where $\tilde{n}\triangleq \bm{w}_{\rx}^{\HT}\bm{N}\bm{w}_s^\ast=(\bm{w}_s\kron \bm{w}_{\rx})^{\HT}\bm{n}\sim\mathcal{CN}(0,\sigma^2\|\bm{w}_s\|_2^2\|\bm{w}_{\rx}\|_2^2)$ is a zero-mean complex-valued normally distributed noise term by the assumption in \eqref{eq:Y}.

\cref{eq:H_tr} can be interpreted as the effective \emph{point spread function} (PSF) of the Tx-Rx array. The PSF is the spatial impulse response of a linear imaging system \cite[p.~20]{goodman1996introduction}, and hence characterizes the resolution and interference suppression capability of the joint Tx-Rx beamformer. There are many well-known approaches for realizing a desired $H_{\txrx}(\theta)$ by jointly choosing the Tx and Rx beampatterns $H_{\tx}(\theta)$ and $H_{\rx}(\theta)$, including co-prime beamforming \cite{vaidyanathan2011sparsesamplers}, and various adaptive methods based on, for instance, maximizing target detection performance or signal-to-interference-plus-noise-ratio (SINR) \cite{bell1993information,chen2009mimo,pillai2011waveform,liu2014joint}. Related Tx-Rx beamformer design problems are also encountered in MIMO communications \cite{palomar2003joint}, as well as joint communications and sensing applications \cite{tsinos2021joint}. 

\paragraph{Sum co-array interpretation}
By \labelcref{eq:A_krao_A_Sigma}, the sum co-array offers a convenient interpretation of \labelcref{eq:H_tr}. In particular, $H_{\txrx}(\theta)=\bm{w}_\Sigma^{\HT}\bm{a}_\Sigma(\theta)$, where $\bm{a}_\Sigma(\theta)\in\mathbb{C}^{N_\Sigma}$ is a sum co-array manifold vector following \eqref{eq:A_Sigma}, and $\bm{w}_\Sigma = \bm{\Upsilon}(\bm{w}_{\tx}\kron \bm{w}_{\rx})$ can be interpreted as a \emph{sum co-array beamforming weight vector}. Vector $\bm{w}_\Sigma$ depends on both the physical Tx and Rx beamforming weights, as well as the redundancy pattern $ \bm{\Upsilon} \in\{0,1\}^{N_\Sigma \times N_{\rx}N_{\tx}}$ of the joint Tx-Rx array geometry. Indeed, choosing unit physical beamforming weights, $\bm{w}_{\xx}=\bm{1}_{N_{\xx}}$, yields sum co-array weights corresponding to the multiplicities of each virtual sensor, $\bm{w}_\Sigma=\bm{\upsilon}_\Sigma$.

\cref{fig:convolutions_nat} shows an example of the beamforming weights and beampattern (magnitude) of both the physical array and sum co-array in case of a ULA and MRA of equivalent aperture. Both arrays have fully overlapping Tx and Rx sensors. The ULA has $N=9$ physical sensors, whereas the MRA has $N=6$. The Tx-Rx beampattern and beamforming weights of the co-array are a Fourier transform pair: the former is the product of the Tx and Rx beampatterns, and the latter is the convolution of the Tx and Rx beamforming weights. Recovering the Tx and Rx beamformers from the Tx-Rx beampattern can also be viewed as a polynomial factorization problem \cite{mitra2010general,martin2010coarray}.
\begin{figure}
		\begin{tabular}{c c c c c  c} 
			\hline
			&Tx&&Rx&&Sum co-array\\
			\hline\hline
			\multirow{3}{*}{\rotatebox[origin=c]{-90}{ULA}}&
			\begin{minipage}{.18\linewidth}
				\begin{center}
					\includegraphics[width=1\linewidth]{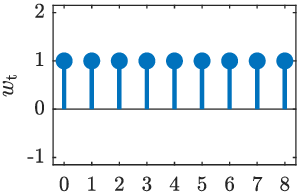} 
				\end{center}
			\end{minipage}
			&$ \ast$ &
			\begin{minipage}{.18\linewidth}
				\begin{center}
					\includegraphics[width=1\linewidth]{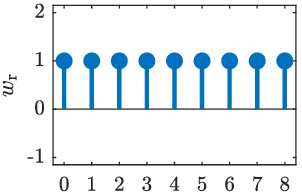} 
				\end{center}
			\end{minipage}
			&$ =  $&
			\begin{minipage}{.36\linewidth}
				\begin{center}
					\includegraphics[width=1\linewidth]{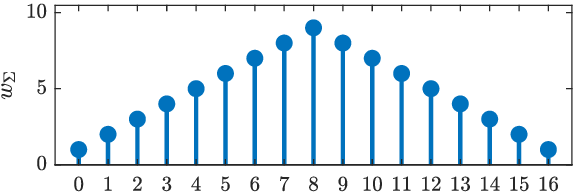} 
				\end{center}
			\end{minipage}\\
			&
			\begin{minipage}{.18\linewidth}
				\begin{center}
					\includegraphics[width=1\linewidth]{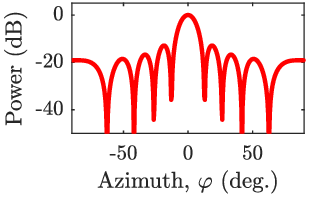} 
				\end{center}
			\end{minipage}
			&$ \times$ &
			\begin{minipage}{.18\linewidth}
				\begin{center}
					\includegraphics[width=1\linewidth]{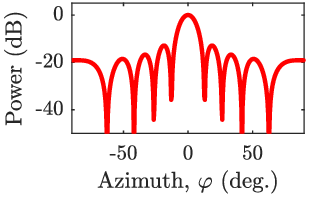} 
				\end{center}
			\end{minipage}
			&$ =  $&
			\begin{minipage}{.36\linewidth}
				\begin{center}
					\includegraphics[width=1\linewidth]{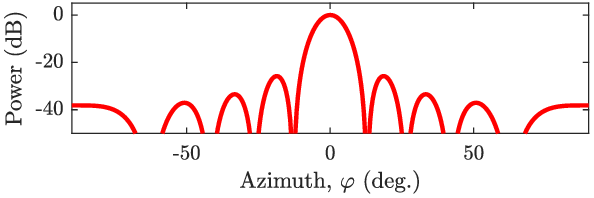} 
				\end{center}
			\end{minipage}\\
			\hline
			\multirow{3}{*}{\rotatebox[origin=c]{-90}{MRA}}&
			\begin{minipage}{.18\linewidth}
				\begin{center}
					\includegraphics[width=1\linewidth]{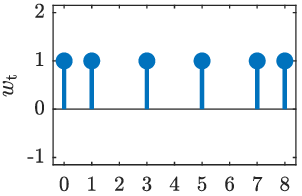} 
				\end{center}
			\end{minipage}
			&$ \ast $ &
			\begin{minipage}{.18\linewidth}
				\begin{center}
					\includegraphics[width=1\linewidth]{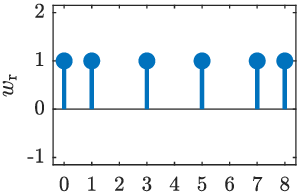} 
				\end{center}
			\end{minipage}
			&$= $&
			\begin{minipage}{.36\linewidth}
				\begin{center}
					\includegraphics[width=1\linewidth]{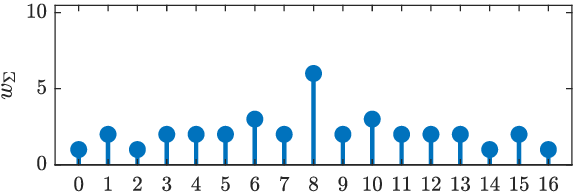} 
				\end{center}
			\end{minipage}
			\\
			&
			\begin{minipage}{.18\linewidth}
				\begin{center}
					\includegraphics[width=1\linewidth]{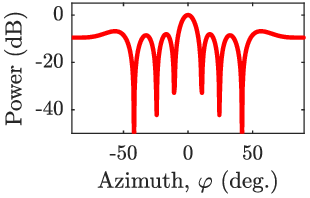} 
				\end{center}
			\end{minipage}
			&$ \times $ &
			\begin{minipage}{.18\linewidth}
				\begin{center}
					\includegraphics[width=1\linewidth]{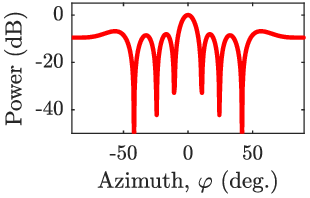} 
				\end{center}
			\end{minipage}
			&$= $&
			\begin{minipage}{.36\linewidth}
				\begin{center}
					\includegraphics[width=1\linewidth]{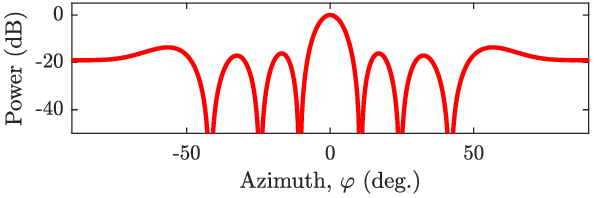} 
				\end{center}
			\end{minipage}
		\end{tabular}
	\caption{Beamforming weights and corresponding beampatterns of ULA and MRA of equivalent aperture. The Tx-Rx (product) beampattern corresponds to a co-array weighting, which is the convolution of the Tx and Rx beamforming weights. Adapted from \cite{rajamaki2021sparsesensor} {\copyright} 2021 R. Rajam\"{a}ki.}\label{fig:convolutions_nat}
\end{figure}

\subsection{Image addition}\label{sec:image_addition}

The beamformer output \eqref{eq:y_txrx} can be written using \eqref{eq:H_tr} as $y=\bm{w}_{\txrx}^{\HT}(\bm{A}_{\tx}\krao \bm{A}_{\rx})\bm{\gamma}+\tilde{n}$, where $\bm{w}_{\txrx}\in\mathbb{C}^{N_{\tx}N_{\rx}}$ denotes the effective Tx-Rx beamforming weight vector
\begin{align}
	\bm{w}_{\txrx}=\bm{w}_{\tx}\kron\bm{w}_{\rx}=\vecm(\bm{w}_{\rx}\bm{w}_{\tx}^{\T}).\label{eq:y_txrx_2}
\end{align}
Hence, $\bm{w}_{\txrx}$ is constrained to vectors in $\mathbb{C}^{N_{\tx}N_{\rx}}$ with Kronecker structure. The idea of \emph{image addition} \cite{hoctor1990theunifying} is to replace rank-1 matrix $\bm{w}_{\rx}\bm{w}_{\tx}^{\T}$ in \eqref{eq:y_txrx_2} by a synthesized \emph{Tx-Rx beamforming weight matrix} $\bm{W}$ of arbitrary rank. Specifically, let $\bm{W}_{\xx}\triangleq [\bm{w}_{\xx}[1],\bm{w}_{\xx}[2],\ldots,\bm{w}_{\xx}[Q]]\in\mathbb{C}^{N_{\xx}\times Q}$ denote a Tx/Rx beamforming weight matrix with $Q\leq \min(N_{\tx},N_{\rx})$ linearly independent columns. Matrix $\bm{W}\in\mathbb{C}^{N_{\rx}N_{\tx}}$ may then be defined as
\begin{align}
	\bm{W}\triangleq  \bm{W}_{\rx}\bm{W}_{\tx}^{\T}=\sum_{q=1}^Q \bm{w}_{\rx}[q]\bm{w}_{\tx}^{\T}[q].\label{eq:W}
\end{align}
We call $Q=\rank(\bm{W})$ the number of (linearly independent) \emph{component images}. 

Synthetic Tx-Rx beamforming vector $\vecm(\bm{W})$ is realized by linearly combining the beamformed outputs of different Tx-Rx beamforming weight pairs. A crucial assumption is that the scattering scene is stationary over the acquisition period in the sense that $\bm{\gamma}$ is fixed and independent of $q$. Indeed, only waveform matrix $\bm{S}$ and noise matrix $\bm{N}$ in \eqref{eq:Y} may vary with $q$. Denoting the Tx-Rx beamformer output of the $q$th component image by $y[q]\triangleq \bm{w}_{\rx}^{\HT}[q]\bm{Y}[q]\bm{w}_s^\ast[q]$, the joint Tx-Rx beamformer output after image addition is defined as
\begin{align}
	\bar{y}\triangleq \sum_{q=1}^Q y[q]
	=\vecm^{\HT}(\bm{W})(\bm{A}_{\tx}\krao \bm{A}_{\rx})\bm{\gamma}+\bar{n},\label{eq:y_txrx_ia}
\end{align}
where $\bar{n}=\sum_{q=1}^Q \bm{w}_{\rx}[q]\bm{N}[q]\bm{w}_s^\ast[q]$ is a noise term. Consequently, the synthesized Tx-Rx beampattern (or synthesized PSF) becomes
\begin{align}
	\bar{H}_{\txrx}(\theta)
	=\sum_{q=1}^Q H_{\txrx}^{(q)}(\theta)
	=\vecm^{\HT}(\bm{W})(\bm{a}_{\tx}(\theta)\kron\bm{a}_{\rx}(\theta)),
	\label{eq:H_tr_ia}
\end{align}
where $H_{\txrx}^{(q)}(\theta)=H_{\tx}^{(q)}(\theta)H_{\rx}^{(q)}(\theta)$; $H_{\tx}^{(q)}(\theta)=\bm{w}_{\tx}^{\HT}[q]\bm{a}_{\tx}(\theta)$; and $H_{\rx}^{(q)}(\theta)=\bm{w}_{\rx}^{\HT}[q]\bm{a}_{\rx}(\theta)$ denote $q$th component (complex-valued) Tx-Rx, Tx, and Rx beampatterns, respectively. 

\paragraph{Sum co-array interpretation}
The synthesized Tx-Rx beampattern in \labelcref{eq:H_tr_ia} can be rewritten using sum co-array manifold vector $\bm{a}_\Sigma(\theta)\in\mathbb{C}^{N_\Sigma}$ as
\begin{align}
	\bar{H}_{\txrx}(\theta)
	= \bm{{w}}_\Sigma^{\HT}\bm{a}_\Sigma(\theta),\label{eq:H_tr_ia_sca}
\end{align}
where $\bm{w}_\Sigma\in\mathbb{C}^{N_\Sigma}$ is the (synthesized) sum co-array beamforming weight vector
\begin{align}
	\bm{{w}}_\Sigma= \bm{\Upsilon}\vecm(\bm{W}). \label{eq:w_sca}
\end{align}

\cref{fig:convolutions_rect} demonstrates image addition in case of the ULA and MRA of equal aperture in \cref{fig:convolutions_nat}. Both array configurations achieve the desired Tx-Rx beampattern due to having equivalent sum co-arrays. The MRA requires $Q=2$ component images as it has fewer physical sensors than the ULA. Next, we briefly discuss how to select $\bm{W}$, and ultimately, $\bm{W}_{\tx}$ and $\bm{W}_{\rx}$ in \eqref{eq:W}.
\begin{figure}
		\begin{tabular}{c c c c c  c} 
			\hline
			&Tx&&Rx&&Sum co-array\\
			\hline\hline
			\multirow{3}{*}{\rotatebox[origin=c]{-90}{ULA}}&
			\begin{minipage}{.18\linewidth}
				\begin{center}
					\includegraphics[width=1\linewidth]{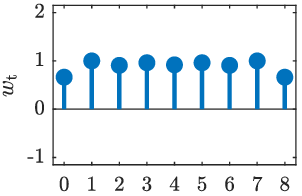} 
				\end{center}
			\end{minipage}
			&$ \ast$ &
			\begin{minipage}{.18\linewidth}
				\begin{center}
					\includegraphics[width=1\linewidth]{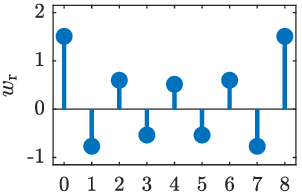} 
				\end{center}
			\end{minipage}
			&$ =  $&
			\begin{minipage}{.36\linewidth}
				\begin{center}
					\includegraphics[width=1\linewidth]{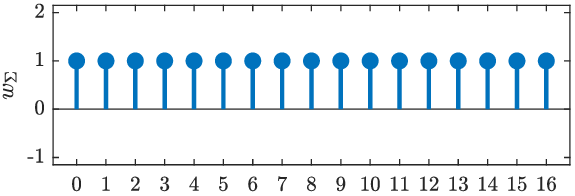} 
				\end{center}
			\end{minipage}\\
			&
			\begin{minipage}{.18\linewidth}
				\begin{center}
					\includegraphics[width=1\linewidth]{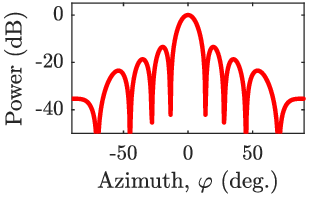} 
				\end{center}
			\end{minipage}
			&$ \times$ &
			\begin{minipage}{.18\linewidth}
				\begin{center}
					\includegraphics[width=1\linewidth]{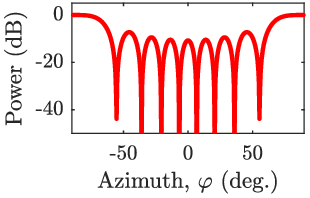} 
				\end{center}
			\end{minipage}
			&$ =  $&
			\begin{minipage}{.36\linewidth}
				\begin{center}
					\includegraphics[width=1\linewidth]{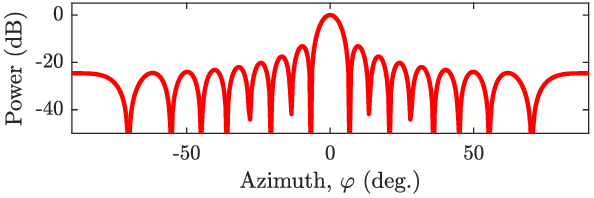} 
				\end{center}
			\end{minipage}\\
			\hline
			\multirow{5}{*}[-5ex]{\rotatebox[origin=c]{-90}{MRA}}&
			\begin{minipage}{.18\linewidth}
				\begin{center}
					\includegraphics[width=1\linewidth]{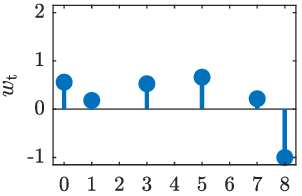} 
				\end{center}
			\end{minipage}
			&$ \ast $ &
			\begin{minipage}{.18\linewidth}
				\begin{center}
					\includegraphics[width=1\linewidth]{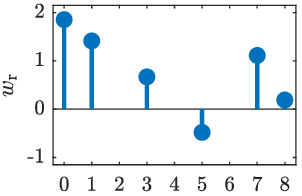} 
				\end{center}
			\end{minipage}
			&\multirow{2}{*}[-2ex]{$\Bigg]\hspace{-7.6pt}\oplus= $}&
			\multirow{2}{*}[-2ex]{
				\begin{minipage}{.36\linewidth}
					\begin{center}
						\includegraphics[width=1\linewidth]{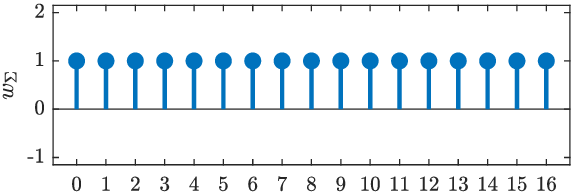} 
					\end{center}
				\end{minipage}
			}\\
			&
			\begin{minipage}{.18\linewidth}
				\begin{center}
					\includegraphics[width=1\linewidth]{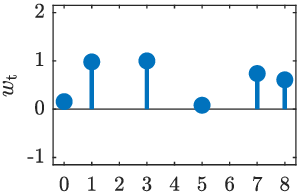} 
				\end{center}
			\end{minipage}
			&$ \ast $ &
			\begin{minipage}{.18\linewidth}
				\begin{center}
					\includegraphics[width=1\linewidth]{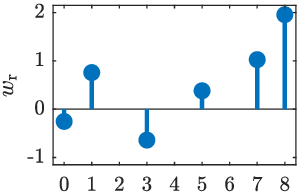} 
				\end{center}
			\end{minipage}
			&&\\
			&
			\begin{minipage}{.18\linewidth}
				\begin{center}
					\includegraphics[width=1\linewidth]{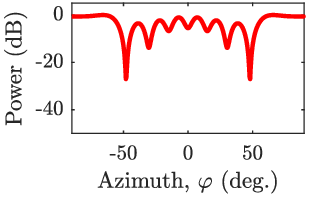} 
				\end{center}
			\end{minipage}
			&$ \times$ &
			\begin{minipage}{.18\linewidth}
				\begin{center}
					\includegraphics[width=1\linewidth]{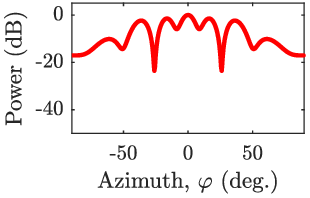} 
				\end{center}
			\end{minipage}
			&\multirow{2}{*}[-2ex]{$\Bigg]\hspace{-7.6pt}\oplus= $}&
			\multirow{2}{*}[-2ex]{
				\begin{minipage}{.36\linewidth}
					\begin{center}
						\includegraphics[width=1\linewidth]{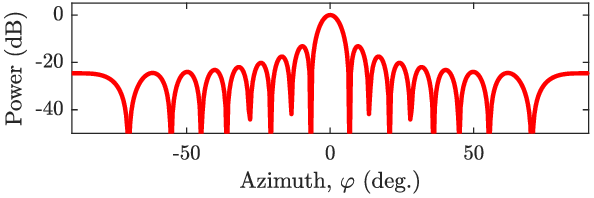} 
					\end{center}
				\end{minipage}
			}\\
			&
			\begin{minipage}{.18\linewidth}
				\begin{center}
					\includegraphics[width=1\linewidth]{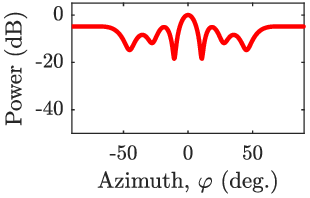} 
				\end{center}
			\end{minipage}
			&$ \times $ &
			\begin{minipage}{.18\linewidth}
				\begin{center}
					\includegraphics[width=1\linewidth]{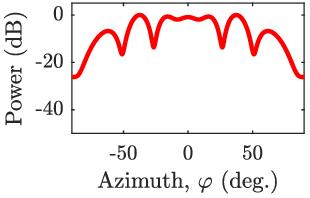} 
				\end{center}
			\end{minipage}
			&&\\
		\end{tabular}
	\caption{Tx-Rx beampattern synthesis using image addition. Both the ULA and MRA achieve the same effective beampattern due to having equivalent sum co-arrays. The MRA requires an additional component image due to having fewer physical sensors than the ULA.}\label{fig:convolutions_rect}
\end{figure}

\subsubsection{Joint optimization of Tx and Rx beamformers}
Designing beamforming weight matrix $\bm{W}$ (given an array geometry) can be cast as the following problem: find $\bm{W}$ synthesizing a desired Tx-Rx beampattern $\bar{H}_{\txrx}(\phi)$ in \eqref{eq:H_tr_ia} for $\phi\in[-\pi/2,\pi/2)$.\footnote{The more challenging problems of only constraining the desired beampattern \emph{magnitude} response $|\bar{H}_{\txrx}(\phi)|^2$, or allowing for a variable array geometry, are not considered herein.} This task, which can be viewed as an FIR (finite impulse response) filter design problem, may be simplified by either discretizing the angular space or, equivalently, by satisfying \eqref{eq:w_sca} for a desired co-array weight vector $\bm{{w}}_\Sigma\in\mathbb{C}^{N_\Sigma}$, provided the desired beampattern $\bar{H}_{\txrx}(\phi)$ is realizable. A key observation is that the system of equations in \eqref{eq:w_sca} is \emph{underdetermined} when the array configuration is \emph{redundant}. That is, if $ N_\Sigma< N_{\rx}N_{\tx} $, then infinitely many $ \bm{W} $ satisfy $ \bm{{w}}_\Sigma=\bm{\Upsilon}\vecm(\bm{W}) $, since $\bm{\Upsilon}\in\{0,1\}^{N_\Sigma\times N_{\tx}N_{\rx}}$ has a non-trivial null space. Therefore, the question arises: how should $\bm{W}$ be chosen? 

One alternative is the least-squares (LS) solution $\vecm(\bm{W})=\bm{\Upsilon}^{\T}\big(\bm{\Upsilon}\bm{\Upsilon}^{\T}\big)^{-1}\bm{{w}}_\Sigma$, which has the intuitive interpretation of distributing the entries of $\bm{{w}}_\Sigma$ \emph{equally} among the associated Tx-Rx sensor pairs \cite{kozick1991linearimaging,kozick1992coarray}. However, the LS solution does not generally yield a low-rank $ \bm{W} $, which implies that unnecessarily many transmissions or linearly independent Tx waveforms may be required for achieving a desired co-array weighting. Hence, one may instead solve a low-rank matrix recovery problem, where $\rank(\bm{W})$ is minimized subject to \eqref{eq:w_sca}. In practice, this problem may be approximately solved by, e.g., relaxing the rank objective to the nuclear norm, or by employing alternating minimization on the rank-revealing form $\bm{W}=\bm{W}_{\rx}\bm{W}_{\tx}^{\T}$ in \eqref{eq:W} combined with bisection over the rank (inner dimension) $Q$ \cite[pp.~85--86]{rajamaki2021sparsesensor}. Rather than satisfying $\bm{{w}}_\Sigma=\bm{\Upsilon}\vecm(\bm{W})$ with equality, an approximation error tolerance can also be allowed in practice \cite{rajamaki2021sparsesensor}.

\paragraph{Factorizing $\bm{W}$}
Given $ \bm{W} $, there is generally no unique choice for physical beamforming weight matrices $\bm{W}_{\tx}$ and $\bm{W}_{\rx}$ in \eqref{eq:W}. However, a convenient factorization of $\bm{W}$ is given by its singular value decomposition (SVD): $ \bm{W}=\bm{U}\bm{\Sigma}\bm{V}^{\HT} $, where $\bm{\Sigma}\in\mathbb{R}_+^{Q\times Q}$ is the diagonal matrix of positive singular values, and $ \bm{U}\in\mathbb{C}^{N_{\rx}\times Q} $, $ \bm{V} \in\mathbb{C}^{N_{\tx}\times Q} $ are the corresponding left-unitary singular vector matrices. Using the SVD, we may then set $ \bm{W}_{\tx} =\bm{V}^\ast$ and $ \bm{W}_{\rx}=\bm{U}\bm{\Sigma}$, for example. The factorization may be more involved in practice due to constraints imposed on $\bm{W}_{\xx}$, such as unit-modulus structure arising from the use of phase shifters \cite{rajamaki2019analog,rajamaki2020hybrid} or phase coded waveforms.

\paragraph{Spatio-temporal trade-offs}
A fundamental question regarding sparse arrays employing beamforming is as follows: Given \emph{any} desired co-array weight vector $\bm{{w}}_\Sigma\in\mathbb{C}^{N_\Sigma}$, how many component images $ Q=\rank(\bm{W}) $ are required to satisfy $\bm{{w}}_\Sigma=\bm{\Upsilon}\vecm(\bm{W})$ exactly? Intuition suggests that the sparser the array, the more component images are necessary to synthesize arbitrary co-array weightings. Indeed, the ULA with fully overlapping Tx and Rx sensors achieves \emph{any} given $\bm{{w}}_\Sigma$ when $ 1\leq Q\leq 2 $ \cite[pp.~86--88]{rajamaki2021sparsesensor}, whereas $ Q\propto N $ is necessary for sparse arrays with $N_{\tx}\propto N_{\rx}\propto N$ physical sensors and $N_\Sigma \propto N^2$ virtual sensors \cite[Proposition~1]{rajamaki2020hybrid}. Naturally, \emph{some} $ \bm{{w}}_\Sigma$ may be achieved either exactly or approximately using a smaller value of $ Q $. For example, $Q=1$ suffices to realize any $\bm{{w}}_\Sigma$ of the form $\bm{{w}}_\Sigma=\bm{\Upsilon}(\bm{t}\kron\bm{r})$, where $\bm{t}\in\mathbb{C}^{N_{\tx}}$ and  $\bm{r}\in\mathbb{C}^{N_{\rx}}$ (i.e., set $\bm{W}=\bm{r}\bm{t}^{\T}$). A direction for future work is investigating practically useful Tx-Rx beampatterns that can be achieved using sparse arrays employing $Q=\mathcal{O}(1)$ component images.

\section{Applications}\label{sec:applications}

We conclude the chapter by highlighting some classical, as well as timely and emerging practical applications of sparse arrays in active sensing. 

\paragraph{Imaging}
Imaging is a key task in applications including radar, sonar, and diagnostic medical ultrasound. An image typically consist of discrete pixels corresponding to a beamformed output such as \eqref{eq:y_txrx_ia}. Consequently, the beampattern synthesis approach of \cref{sec:coarray_beamforming} is well-suited for imaging using sparse arrays. Beamforming-based imaging allows extracting useful information from complex scattering environment containing continuous or distributed scatterers, potentially located in the near-field of the array. For example, in medical ultrasound imaging, the ultrasonic probe is in direct contact with the skin and hence in close-proximity to internal organs and other scatterers with intricate shapes. Despite real-world challenges, both linear and planar sparse arrays geometries employing beamforming have been successfully applied in practice \cite{kozick1993synthetic,lockwood1998realtime3d,ahmad2004designandimplementation,cohen2018sparseconvolutional,cohen2021sparse,lehtonen2021medical}. Sparse array imaging techniques can also offset beampattern degradation due to constraints on the beamforming architecture, such as unit-modulus beamforming weights or coarse phase shifter quantization \cite{rajamaki2020hybrid}.

\paragraph{MIMO radar}
MIMO radar systems frequently employ sparse arrays, such as the nested geometry in \eqref{eq:MRA_no_nst}, to improve angular resolution and target identifiability \cite{li2007mimoradar}. Topical applications of MIMO radar include autonomous sensing \cite{hugler2018radar} and, in particular, automotive radar \cite{patole2017automotive,sun2020mimoradar,engels2021automotive}. Sparse arrays are attractive in such applications since they enable inexpensive, compact, and mobile system architectures. Moreover, as achieving high angular resolution given a limited sensor budget is often critical, it may be necessary to extend the physical aperture to a degree where the co-array is no longer contiguous. In such cases, array interpolation techniques can reap the benefits of a large contiguous (interpolated) co-array, while mitigating ambiguities arising from the actual nonuniform sampling \cite{qiao2017unified,sun2020asparselinear,sun20214dautomotive,sarangi2022single}.

\paragraph{Wireless communications}
An application of sparse arrays in wireless communications is channel estimation. Especially at high frequencies, communication channels typically exhibit spatial sparsity with only a few strong line-of-sight components \cite{ayach2014spatially}. Estimating some of the channel parameters, such as directions of departure (arrival) at the transmitter (receiver), may thereby benefit from the improved resolution or identifiability offered by sparse arrays. At sub-millimeter wavelengths, hybrid or fully analog beamforming architectures may have to be employed for reasons of cost and power consumption; emphasizing the importance of beamformer design also in channel estimation \cite{shahsavari2022beamspace}. Sparse array configurations and sampling techniques can be used in conjunction with such beamforming architectures to further save system resources in, e.g., channel subspace estimation \cite{haghighatshoar2017massive} or direction finding in general \cite{ibrahim2017design,guo2018doaestimation,koochakzadeh2020compressed}. Sensing is also envisioned to play an increasingly important role in 6G and beyond communications \cite{saad2020avisionof6g}. Such emerging joint communication and sensing systems \cite{mishra2019toward,ma2020joint,ahmadipour2022aninformation} present a plethora of yet largely unexplored opportunities for sparse arrays.

\section{Conclusions}\label{sec:conclusions}
This chapter gave an overview of active sensing using sparse arrays, with a focus on low-redundancy one-dimensional array configurations and angular-domain transmit-receive (Tx-Rx) beamforming. A simple but general signal model was described highlighting the ubiquity of the sum co-array in active sensing. A review was provided of the active Minimum-redundancy array (MRA) for various degrees of shared sensors between the Tx and Rx arrays. In the case when all Tx and Rx sensors are shared, symmetric array configurations were introduced that are scalable, meaning that they can be generated for any number of sensors $N$. In contrast, the corresponding MRA is computationally challenging to find and hence still unknown for $N\geq 49$. Interestingly, however, there exist symmetric MRAs for each $N\leq 48$, suggesting the possibility of an optimal class of symmetric arrays. Properly designed symmetric arrays have both contiguous sum and difference co-arrays, which makes them directly applicable to both active and passive sensing. Moreover, extending the set of achievable Tx-Rx beampatterns by image addition was considered. Synthesizing a desired Tx-Rx beampattern can be interpreted as adjusting the virtual beamforming weights of the sum co-array. When the array configuration contains redundant virtual sensors, the physical beamforming weights may be optimized to minimize the beampattern synthesis (acquisition) time. This suggests a spatio-temporal trade-off between array sparsity and the number of physical beamforming weight vectors, or ``component images'', required in the synthesis. Indeed, the sparser the array, the more component images are necessary to achieve arbitrary co-array beamforming weights. Finally, selected timely and emerging active sensing applications were discussed. Sparse arrays offer cost-efficiency and improved performance in many use cases.

\backmatter

	\backmatter
	
	\bibliographystyle{abbrv}
	\bibliography{../references.bib}
	
	\printindex
	
\end{document}